\documentclass{IEEEtran}

\usepackage{subcaption}
\usepackage{graphicx}
\usepackage[nolist]{acronym}
\usepackage{xspace}
\usepackage{url}

\newtheorem{assumption}{Assumption}
\newtheorem{theorem}{Theorem}

\newtheorem{definition}{Definition}
\newtheorem{remark}{Remark}
\newtheorem{correctness-argument}{Argument}
\newtheorem{lemma}{Lemma}

\newcommand{\fakedescription}[1]{\medskip\noindent\textit{\textbf{#1}}}

\newcommand{\mysystem}{{\sc Lara}\xspace}

\begin{document}

\begin{acronym}
  \acro{GS}{Group Signatures}
  \acro{RL}{Revocation List}
  \acro{VLR}{verifier local revocation.}
  \acro{CRL}{certification revocation list}
  \acro{VANETs}{vehicular networks}
\end{acronym}

\title{\mysystem: Lightweight Anonymous Authentication with Asynchronous Revocation Auditability}

\author{Claudio Correia, Guilherme Santos, Lu\'{\i}s Rodrigues \\ INESC-ID \\Instituto Superior Técnico, \\Universidade de Lisboa}

\date{\today}

\maketitle

\begin{abstract}
Anonymous authentication is a technique that allows to combine access control with privacy preservation. Typically, clients use different pseudonyms for each access, hindering providers from correlating their activities. To perform the revocation of pseudonyms in  a privacy preserving manner is notoriously challenging. When multiple pseudonyms are revoked together, an adversary may infer that these pseudonyms belong to the same client and perform privacy breaking correlations, in particular if these pseudonyms have already been used. Backward unlinkability and revocation auditability are two properties that address this problem. Most systems that offer these properties rely on some sort of time slots, which assume a common reference of time that must be shared among clients and providers; for instance, the client must be aware that it should not use a pseudonym after a certain time or should be able to assess the freshness of a revocation list prior to perform authentication. In this paper we propose  \mysystem, a Lightweight Anonymous Authentication with Asynchronous Revocation Auditability that does not require parties to agree on the current time slot and it is not affected by the clock skew. Prior to disclosing a pseudonym, clients are provided with a revocation list (RL) and can check that the pseudonym has not been revoked. Then, they provide a proof on non-revocation that cannot be used against any other (past or future) RL, avoiding any dependency of timing assumptions.  \mysystem can be implemented using efficient public-key primitives and space-efficient data structures. We have implemented a prototype of \mysystem and have assessed experimentally its efficiency.
\end{abstract}

\section{Motivation and Goals}
\label{sec:intro}

Authentication is a fundamental requirement in many applications, ensuring that only authorized users can access protected resources. However, authentication processes often pose privacy risks, as they may involve sensitive information, such as habits or locations~\cite{pan2013crowd}. These concerns are particularly relevant in edge applications, including crowdsensing~\cite{ganti2011mobile,pan2013crowd,ni2017security,bastos2018signature} and vehicular networks (VANETs)\cite{remeli2019automatic,ganan2015epa,mixzone}, where mobile clients must authenticate frequently—such as when transitioning between base stations or cells—leading to the continuous generation of sensitive data. 
Information from multiple authentications may be linked to extract additional information such as daily routines~\cite{ganti2011mobile} or health conditions~\cite{lin2012bewell} for financial gain~\cite{GDPR,christin2016privacy,lauinger2012privacy,at_t,gao_17_656,monetising_car_data}.

Anonymous authentication offers both accountability and privacy, protecting clients from curious application providers while ensuring that only authorized participants are able to use the application~\cite{rahaman2017provably,ishida2018fully,khodaei2018efficient,sun2010efficient}. Anonymous authentication can be achieved using \ac{GS} schemes~\cite{chaum1991group,boneh2004group,bastos2018signature} or pseudonym certificates~\cite{lysyanskaya1999pseudonym,chaum1985security}. A desirable feature of anonymous authentication system is the ability to support \emph{revocation} without violating privacy. Revocation aims to prevent some clients from further authenticating in the system. Client revocation may be required in the event of credential misuse, a change in client privileges, stolen secret keys, or simply when a client leaves voluntarily. 

Client revocation can be implemented in different ways. We distinguish two main classes of revocation strategies, namely, \emph{global client revocation} and \emph{verifier local revocation}. Strategies based on global client revocation require all clients to obtain new credentials (or update their credentials) every time a single client is revoked. Examples of this strategy include Ateniese \textit{et al.}~\cite{ateniese2002quasi} (where the group public key is renewed at each revocation) and Ohara \textit{et al.}~\cite{ohara2019shortening} (where a small public membership message is broadcast at each revocation). These approaches make revocation very onerous in scenarios with many clients (e.g., consider vehicle numbers in VANETs) and impractical in mobile settings, where clients may become temporarily disconnected from the network.  Strategies based on \ac{VLR}~\cite{boneh2004group,bringer2011backward} do not require that all clients are contacted when a given client is revoked. Instead, only the nodes that perform authentication (often called the signature \emph{verifiers}) have to be informed about the revoked clients~\cite{Vtoken,vespa,khodaei2018secmace,sun2010efficient,ishida2018fully,bastos2018signature}. In systems that use pseudonyms, this involves sending to the verifiers a \ac{RL} with the pseudonyms of the revoked client. In systems based on group signatures, this involves sending a cryptographic token that can be used to trace the digital signatures of the revoked client. In this paper, we are interested in systems that support verifier local revocation.

A challenge in revocation strategies is that, if one or more credentials have been used before revocation, an attacker can cross-check the information used for revocation with the information collected when those pseudonyms were used to break the privacy of the client. In order to respect \emph{backward unlinkability}~\cite{haas2011efficient,khodaei2018efficient,ishida2018fully,nakanishi2005verifier}, client revocation should not allow linking credentials that have been used prior to the revocation. Previous strategies to provide backward unlinkability assign credentials that are valid only during a given time slot of a certain duration~\cite{haas2011efficient,khodaei2018efficient,rahaman2017provably,bastos2018signature}. Recent work has shown that is possible to make time slots very small~\cite{paper_CCS}, but do not eliminate completely the tension between immediate revocation and backward unlikability. Most critically, these approaches assume a common time reference among all participants: pseudonyms are bound and revoked based on time slots and credentials are tied to specific time slots. If clocks are not synchronized, authentication may fail because the client and the verifier may be operating in distinct time slots.

Given that many revocation schemes cannot ensure privacy for clients that have been revoked, another desirable property of an anonymous authentication  system is \textit{revocation auditability}~\cite{henry2011formalizing}, that states that a client should be able to see his revocation status, before each authentication, avoiding the situation where a client tries to authenticate with a pseudonym that has been revoked without being aware of that fact. An example of a system that offers revocation auditability is Nymble~\cite{tsang2009nymble}.  In Nymble, a client is provided with a revocation list (RL) prior performing authentication and can check is revocation status before proceeding with the operation; if it finds that the pseudonym it was planning to use has been revoked, it will not disclose it to the verifier. Similarly to the systems mentioned in the previous paragraph, Nymble also relies on the notion of timeslots.

In this paper, we propose \mysystem, a \underline{L}ightweight Anonymous Authentication with \underline{A}synchronous \underline{R}evocation \underline{A}uditability that does not rely on a shared notion of time among clients and verifiers. Prior to disclosing a pseudonym, clients are provided with a revocation list (RL) and can check that the pseudonym has not been revoked, \textit{in that specific RL} (regardless of its status on other RLs). Then, clients provide a proof on non-revocation that works with the given RL but that cannot be used against any other (past or future) RL, avoiding any dependency of timing assumptions.  Besides offering revocation auditability, \mysystem offers strong backward (and forward) unlinkability properties that are not tied to time intervals  (unlike many previous solutions~\cite{haas2011efficient,khodaei2018efficient,rahaman2017provably,bastos2018signature,paper_CCS}). This feature allows \mysystem to achieve immediate revocation for clients by publishing RLs without any time interval delay, while also requiring smaller RLs. \mysystem can be implemented using efficient public-key primitives and space-efficient data structures. We propose and compare 3 alternative implementations of \mysystem that use Bloom filters (in different ways) to construct the revocation list. These implementations allow clients to perform the audit efficiently.
We have implemented a prototype of \mysystem and have experimentally evaluated its efficiency.

\section{Revocation Auditability}
\label{sec:relatedwork}

Despite the existence of various schemes for anonymous authentication~\cite{paper_CCS,henry2011formalizing}, they face a significant challenge in maintaining client anonymity after their revocation~\cite{henry2011formalizing}. This issue arises in authentication systems that respect verifier local revocation (VLR)~\cite{boneh2004group,bringer2011backward}, as it requires revocation lists to be published. The existence of these lists allows an attacker to flag/identify the authentication instances performed by a revoked client, which can be exploited to link multiple authentications and, consequently, compromise the client's anonymity. An adversary can exploit these lists to link past and future authentications. Backward unlinkability is the property that safeguards past authentication anonymity. Revocation auditability, a key goal of our system, aims to provide unlinkability in future authentications performed by a revoked user. 

Informally, revocation auditability means that a user has the ability to verify his revocation status at a service provider before attempting to authenticate~\cite{henry2011formalizing,tsang2009nymble}. If the user is indeed revoked, he can then safely disconnect from the service without disclosing any potentially sensitive information. Otherwise, clients could be revoked without their knowledge, and a malicious service provider might still accept authentication requests from these revoked users, thereby compromising their privacy. Such an attack can lead to severe privacy breaches, enabling the service provider to link all of the user's actions, which is especially concerning when schemes rely on uncircumventable forward linkability for achieving revocability. In the literature on anonymous authentication we have identified four main techniques to achieve revocation auditablity:

\fakedescription{Central authority:} The simplest solution is to have clients access a central and remote node that issues revocation lists. Before each authentication, clients pull the most recent list and verify their status in the system (i.e., whether they have been revoked). However, this solution is not desirable for large-scale and dynamic systems like distributed edge storage systems or VANETs, where clients have intermittent connections, and reliance on a central node can lead to multiple availability failures. Additionally, it does not respect VLR.

\fakedescription{Contract-based revocation:} Another approach is to use contract-based revocation~\cite{henry2011formalizing,schwartz2009contractual,loesing2009measuring}, where the contract semantics are agreed upon by both the user and the provider. This enables the user to determine whether a certain action will constitute misbehavior before deciding whether to engage in it. Thus, the client is aware that it may be revoked after such an action. Unfortunately, due to the large variety of applications nowadays, it is very difficult to define all the possible behaviors a client may exhibit, making these approaches inflexible and impractical.

\fakedescription{Revocation list freshness:} A more desirable approach is to ensure that fresh revocation information reaches the client. This is achieved by having RLs published at regular $\Delta$\textit{t} time intervals, containing a signature with the corresponding timestamp~\cite{tsang2009nymble}. When a client performs authentication, it can first request the local revocation list, which must have a fresh signature for the current $\Delta$\textit{t}, and then check if it has not been revoked; otherwise, it should halt the authentication process. This is a practical and easily deployable solution, but the downside is that $\Delta$\textit{t} imposes a tradeoff between system availability (clients do not authenticate if revocation information is not fresh) and effective revocation (the larger the $\Delta$\textit{t}, the longer it takes for a revocation to take effect). Maybe more concerning, the availability of some of these systems relies on the assumption that participants have their clocks synchronized (given that authentication may fail if the client and the verifier do not agree on the current time slot), which can be a vulnerability.

\fakedescription{Non-Revocation proof based on the RL:} The more secure approach is to have clients locally generate a non-revocation proof unique to the presented RL (before authentication, the client downloads the list from the local provider). This guarantees that the generated proof cannot be tested against another RL (that may contain the client), and it is only valid for the locally presented RL. Previous approaches to offer revocation audidatability based on non-revocation proofs have been designed for systems that do not rely on a Trusted Third Party (TTP) to implement revocation. Although powerful, these solutions rely on the use of NZNPs~\cite{au2008perea,BLAC_revoking} that require clients to construct complex proofs, in the critical path of the authentication procedure, imposing a high latency and computation cost. Therefore, these solutions are unnecessarily expensive when a TTP is available.

\fakedescription{Our goals:} Our goal is to derive an anonymous authentication for existing edge applications, such as vehicular networks, where a TTP is responsible for issuing and revoking pseudonyms. We aim at offering the following combination of features that are desirable in this setting:  immediate revocation (there is no need to wait for the end of some pre-defined time-slot to revoke a client), verifier local revocation (edge resources can perform authentication without being required to contact the TTP in the critical path),  backwards unlinkability, revocation auditability, robustness in face of the clock-skew (the safety and/or liveness of the algorithm does not depend on the clock synchronization of all agents), and, last but not the least, can be implemented efficiently.

\section{System Model} 
\label{sec:systemModel}

\subsection{Entities}
\label{sec:entities}

The system is composed of the following entities: clients, verifiers, a (logically) centralized pseudonym manager service, and a trusted administrator. We follow a nomenclature similar to that in previous work~\cite{rahaman2017provably,haas2011efficient}.

\fakedescription{Clients:} the application client that generates signatures to perform authentication against any verifier. Clients are the holders of pseudonyms that they use to generate capabilities to ensure anonymity. Clients are responsible for renewing their pseudonyms when needed. 

\fakedescription{Verifiers:} the component that performs client authentication before granting access to a resource, such as edge storage. Verifiers are responsible for checking the validity of the pseudonyms provided by clients before granting access. They are also responsible for updating their state by fetching the list of revoked capabilities from the pseudonym manager.

\fakedescription{Certification Authority (CA):} this component is responsible for providing new pseudonyms to clients and, when necessary, revoking capabilities generated from these pseudonyms. CA servers are the only entity capable of accessing the true identity of a client.

\fakedescription{Administrator:} a trusted entity responsible for adding clients to the system and instructing the CA to revoke clients. 

\subsection{Timing Assumptions}
\label{sec:timing}

We assume that participants have access to loosely synchronized clocks and can assess whether certificates have expired. As will be explained in Section~\ref{sec-pk:implementation}, we assume that pseudonyms are valid for extended time periods, denoted \textit{ epochs} (for instance, valid for periods of 1 year). 

However, unlike previous authentication schemes that are based on time slots, our authentication scheme does not require the client and the verifier to agree on the current time slot. For example, in~\cite{paper_CCS}, authentication may fail if the client and verifier clocks are not synchronized: The client will provide credentials that are valid for a time window that differs from the verifier time window. This problem is exacerbated when time slots are small, and, therefore, in those solutions, the tolerance to clock skew conflicts with the goal of reducing the linkability window. Our solution does not suffer from this limitation, since the unlikability properties are not tied to the notion of time.

\subsection{Threat Model}
\label{sec:ThreatModel}

We trust the administrator and the CA.  Clients and verifiers are considered untrusted and susceptible to the control of attackers, potentially engaging in malicious activities. 

\fakedescription{Malicious Client:} may attempt to generate pseudonyms or capabilities to impersonate a valid client and access resources to which it is not authorized. It can also try to use old or fake pseudonyms after being revoked to authenticate towards verifiers. 

\fakedescription{Malicious Verifier:} The problem we consider is that a malicious verifier may try to perform link attacks ~\cite{Vtoken,haas2011efficient}. Such attacks involve associating (linking) various pseudonyms with a single client, thereby compromising user anonymity. This type of attack becomes trivial when revocation lists are disclosed that enumerate all of a client's pseudonyms. A malicious verifier could potentially compile all the observed data with the aim of deducing user identities. 

\fakedescription{Trust Assumptions:} Entities use asymmetric key pairs to establish secure channels. Clients employ pseudonyms for authentication, integrity, and non-repudiation.  The CA will only revoke users if instructed by the trusted and authenticated administrator, and will generate fresh pseudonyms for non-revoked and authenticated clients. We assume that there is no collusion between the trusted CA and the verifiers. In scenarios were potential collusion between the trusted CA and the verifiers is a threat, the architecture can be augmented with additional components, as described in the literature~\cite{tsang2009nymble, Vtoken}. For instance, in Nymble, there is a Certification Authority that knows the real identity of the user which independent of the server that provides the pseudonyms used in the authentication.

\section{\mysystem}
\label{secshort:system}

We now introduce our scheme for anonymous authentication. To the best of our knowledge, \mysystem is the first pseudonym-based system to offer both immediate revocation, auditability and pure backward unlinkability. We use the term ``pure'' because, unlike~\cite{haas2011efficient,paper_CCS}, in \mysystem unlikability is not tied to time slots of discrete granularity.  Our solution enables immediate revocation since a new \ac{RL} can be published without delay, while systems based on time slots must wait for the current slot to end.

\begin{figure}[t]
  \centering
  \includegraphics[width=\columnwidth]{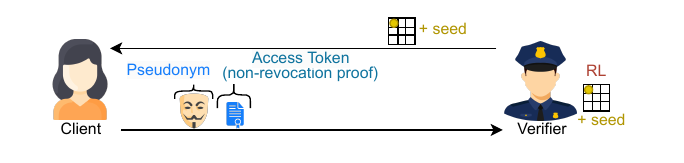} 
  \caption{Authentication in \mysystem} 
  \label{fig:privacyKeeper}
\end{figure}

In \mysystem, the authentication is performed against a Revocation List (RL). A RL includes a unique seed and an encoded set of revoked access tokens. Revoked access tokens are encoded in an RL using a secure one-way function: it is possible to check if a token is in the set but impossible to extract tokens from the set. The authentication
requires the client to transmit a pseudonym and an access token that is not included in the set of revoked access tokens. This token proves that the given pseudonym was not revoked; this is conceptually similar to the non-revocation proofs used in NZNPs~\cite{au2008perea,BLAC_revoking}, as illustrated in Fig.~\ref{fig:privacyKeeper}. By using access tokens, our scheme enables client revocation without disclosing pseudonyms.  To achieve auditability, we require tokens to be uniquely linked to a given RL, similarly to the non-revocation proofs used in NZNPs. To accomplish this, clients request a RL from the verifier at the time of authentication and generate their token exclusively for that RL. Our scheme achieves this by letting access tokens for a given RL be a function of the unique seed of that RL. Clients then generate the token using this seed, establishing a unique connection between the token and the presented RL, as shown in Fig.~\ref{fig:privacyKeeper}. The verifier is then responsible for cryptographically checking that all the presented information has been correctly constructed, for instance, verifying that the presented token corresponds to the RL seed. 

\subsection{Algorithm}
\label{sec-pk:implementation}

We will now describe in more detail how the tokens used for authentication and revocation are constructed. Table~\ref{tab:notation} provides a summary of the notation used in the description of the algorithm.

\begin{table*}[ht]
    \centering
    \caption{Summary of Notation}
    \label{tab:notation}
    \begin{tabular}{lr}
    \hline\hline
         $\langle K^+_{CA}, K^-_{CA} \rangle$  & asymmetric key pair for the CA\\
         $p_x = \langle K^+_{x} , \{K^+_{x}\}^{K^-_{CA}} \rangle$ & pseudonym $x$ \\
         $P_c$ & set of pseudonyms of a client $c$\\
         $\textit{seed}_{z}$ & unique seed of a revocation list\\
         $\sigma_x^{RL_z} = \{ \textit{digest}(seed_{z}) \}^{K^-_{p_{x}}}$ & access token for $p_x$ in revocation list $\textit{RL}_z$\\
         $\textit{digest}(\sigma_x^{RL_z})$ & encoded access token\\
         $\textit{sreat}_z$ &  \underline{s}et of \underline{r}evoked \underline{e}ncoded \underline{a}ccess \underline{t}okens in revocation list $\textit{RL}_z$\\
         $RL_z = \langle \{\textit{seed}_{z}, \textit{sreat}_z, \textit{digest}( seed_z \mathbin\Vert \textit{sreat}_z) \}^{K^-_{CA}} \rangle$ & revocation list\\
         \hline\hline
    \end{tabular}
\end{table*}

\fakedescription{Notation:} We assume $\langle K^+, K^- \rangle$ is a pair of public and private asymmetric keys, respectively. A digital signature is defined by $\{\textit{digest}(a)\}^{K^-}$, where the private key $K^-$ is used to sign the digest of the content of $a$.

\fakedescription{Time:} In this scheme, time is divided into epoch periods with a duration of $e$. These epochs are relatively long time intervals, such as one year or one month. Each pseudonym is associated with a specific epoch period. Epochs are use to simplify the garbage collection of old pseudonyms.

\fakedescription{Client Storage:} We assume that a client has performed some initial setup and stores locally multiple pseudonyms valid for the current epoch. Each pseudonym $p$ has a public and private key pair, $\langle K^+_{p}, K^-_{p} \rangle$, and the pseudonym is defined by its public key and a digital signature from a CA,  $ p = \langle K^+_{p} , \{K^+_{p}\}^{K^-_{CA}} \rangle$. The client holds a number $I$ of valid pseudonyms for the current epoch. 

\fakedescription{Access Token:} An access token $\sigma_x^{RL_z}$ in \mysystem is always cryptographically bounded to a pseudonym $p_x= \langle K^+_{p_x}, \{K^+_{p_x}\}^{K^-_{CA}} \rangle, x\in [1, I]$ and a unique $seed_{z}$. The token is used to proof that some selected pseudonym, $p_x$, has not been revoked in the given revocation list $RL_z$. The next paragraph explains how $RL_z$ is uniquely bounded to $seed_{z}$. The construction of our token is quite straightforward. It involves generating a digital signature using the private key associated with $p_x$ over the digest of the provided $seed_{z}$. Therefore, $\sigma_x^{RL_z}$ is defined as $\{ \textit{digest}(seed_{z}) \}^{K^-_{p_{x}}}$. For authentication, the client forwards both the selected pseudonym and the corresponding token, denoted as $\langle p_x , \sigma_x^{RL_z} \rangle$, as depicted in Fig.~\ref{fig:privacyKeeper}.

\fakedescription{Revocation Lists:} In \mysystem, the RLs include an unique seed and an encoded set of revoked access tokens. Each RL is constructed by a trusted entity, the CA. This CA has the capability to generate all the access tokens that can be tied to any given RL. The generation of RLs is the only aspect of our scheme that has some computational impact because each RL is dependent on a distinct seed. When the CA needs to revoke a client at a specific time $t$, it first generates a random value, which becomes the $seed_{z}$ associated with a new revocation list, $RL_z$. It then retrieves each pseudonym of the recently revoked client, along with all other pseudonyms that were revoked in the current epoch, from other previously revoked clients. Subsequently, it calculates the digest of the seed and generates each token using all the revoked pseudonyms, resulting in $\langle 
 \{\textit{digest}(seed_z)\}^{K^-_{p_i}} \rangle, \forall i \in R_{set}$, where $R_{set}$ is the set of all revoked pseudonyms in the current epoch. Each of these tokens is encoded in the RL's set of revoked access tokens using a secure one-way function. We denote the resulting \underline{s}et of \underline{r}evoked \underline{e}ncoded \underline{a}ccess \underline{t}okens \textit{sreat}$_z$. Finally, the CA creates a final digital signature covering both the RL and the seed together, denoted as $\{ \textit{digest}( seed_z\mathbin\Vert \textit{sreat}_z ) \}^{K^-_{CA}}$. The CA can then \emph{immediately publish} these three elements: the seed, the set of revoked encoded access tokens, and the signature, all of which define the new RL.

In our scheme we ensure the verifier cannot provide a fake $RL_{z}$ by leveraging digital signatures for authenticity. However, similar to related work, a malicious verifier might provide an old RL, misleading a revoked client into a successful authentication attempt. \mysystem does not suffer from this vulnerability and preserves revocation auditability. Our scheme uniquely binds the authentication information with the provided RL (through the seed), rendering it useless to test across different or newer RL versions, regardless of the user's revocation status. Next, we explain how our unique token can still provide proof of non-revocation.

\fakedescription{Non-Revocation \textit{proof}:} For authentication, the client must provide a pseudonym and a non-revocation proof that this pseudonym is valid. The proof is provided by presenting a valid access token that is not included in the RL’s encoded set of revoked access tokens. This provided token is associated with the given pseudonym \emph{and} to the unique seed of the RL. In detail, the process to demonstrate that the given pseudonym $p_x$ is not revoked is as follows.

In the first step, the client must obtain a revocation list $RL_{z}$ and its associated $seed_{z}$ from the verifier. Subsequently, the client validates the signature, generated by the CA, on the revocation list. This validation confirms the integrity and authenticity of both the revocation list and the seed. After this verification, the client generates a token $\sigma_x^{RL_z}$ for a selected pseudonym $p_x$, as described earlier. It is important to note that this access token is only applicable for testing against $RL_{z}$ and it is not valid to be tested against any other RL. Additionally, as previously mentioned, the verifier cannot provide a fake $RL_{z}$.

Consequently, it is the verifier's responsibility to validate $\sigma_x^{RL_z}$ by ensuring that it has been correctly constructed, that it corresponds to the presented pseudonym $p_x$, and that it has not yet been revoked. To achieve this, the verifier first checks if the pseudonym has a valid signature from the CA. Subsequently, it uses the provided pseudonym's public key, $K^+_{p_x}$, to verify the digital signature within $\sigma_x^{RL_z}$. This signature must be correctly constructed using $K^-_{p_x}$ and must correspond to the correct $seed_{z}$ from $RL_{z}$. If this is confirmed, it indicates that the proof has been correctly constructed and corresponds to the presented pseudonym $p_x$.
The next step for the verifier is to ascertain whether the access token $\sigma_x^{RL_z}$ can be found in the encoded set of revoked access tokens of  $RL_{z}$. If this access token belongs to the encoded set of revoked access tokens, the pseudonym has been revoked; otherwise, the access token provides proof that this pseudonym has not been revoked and is valid.

\fakedescription{Authentication:} At a high level, authentication begins with the client downloading the RL from the local verifier. Then, the client selects a pseudonym and generates an access token for the seed associated with RL, using the private key of this pseudonym. Next, the client checks whether it has been revoked by testing if the token is found in the RL's encoded set of revoked access tokens. If not, both the access token and pseudonym are sent to the verifier, which authenticates the token and tests it against the same RL. If the token is found in the RL's encoded set of revoked access tokens, the client is considered revoked; otherwise, it is deemed valid, and the authentication is accepted. The step of downloading the list is necessary to ensure auditability, enabling the client to verify his status. It is important to note that even if the verifier presents an older RL, the proof generated by the client is only valid for that specific RL, thus guaranteeing revocation auditability.

\section{Security Proof}

In this section, we first present a proof that \mysystem preserves unlinkability, and afterward we prove that \mysystem is also capable of offering auditability.

\subsection{Unlikability}

We now provide a proof that \mysystem offers full unlinkability, meaning that revocation information cannot be linked to the information used by clients when authenticating before and after the revocation. Our scheme provides unlinkability because access tokens are only valid for a specific RL and, therefore, the revoked access tokens that are encoded in different RLs are necessarily distinct.

\begin{assumption}
\label{PK-assumption:function}
There is a secure one-way function.
\end{assumption}

We assume the availability of a secure collision-resistant hash function \textsc{h()}, such as SHA256, that is easy to compute on every input, but not possible to invert given the output.

\begin{assumption}
\label{PK-assumption:pseu}
Only the client and the CA can access pseudonyms.
\end{assumption}

We assume that both the client and the CA will securely store the secret private key of each pseudonym and never disclose it. The CA is responsible for generating different pseudonyms for a client. The cryptographic keys $\langle K^- , K^+ \rangle$ of a pseudonym are randomly generated by the CA using a secure cryptographic scheme.

\begin{assumption}
\label{PK-assumption:asymmetric}
There is a secure deterministic digital signature scheme.
\end{assumption}

We assume the availability of a secure deterministic digital signature scheme, such as Ed25519~\cite{bernstein2012high}, which ensures the usual authentication, integrity, and non-repudiation properties. A deterministic signature is generated by $\{\textsc{h}(a)\}^{K^-}$ using $K^-$ to sign the digest of the content of $a$, and can be verified with $K^+$.

\begin{assumption}
\label{PK-assumption:unused}
The client picks an unused pseudonym to authenticate.
\end{assumption}

For each authentication, the client chooses a pseudonym that has never been used. If it runs out of pseudonyms to use, it requests more from the CA.

\begin{assumption}
\label{PK-assumption:revoked}
Given a revocation list $RL_z$, the client only proceeds with authentication with a pseudonym $p_x$ if $p_x$ has not been added to the $RL_z$'s encoded set of revoked access tokens.
\end{assumption}

When performing authentication, the client requests the current $RL_z$ from the verifier, and checks if pseudonym $p_x$ has been revoked. We recall that $p_x$ is marked as revoked in $RL_z$ if $\sigma_x^{RL_z}$ belongs to the RL's encoded set of revoked access tokens. If the client finds that $p$ has been revoked in RL, it stops from using the service and does not provide any more information; otherwise, it continues with the authentication.

\begin{definition}
\label{PK-definition:pseudonym}
A valid pseudonym $p_x$ is defined as:
\[ p_x = \langle K_{x}^-, K_{x}^+, \{\textsc{h}(K_{x}^+)\}^{K_{CA}^-}  \rangle \]
\end{definition}

Where $(K_x^-, K_x^+)$ represent the private and public key of the pseudonyms respectively, and $\{\textsc{h}(K_{x}^+)\}^{K_{CA}^-}$ is a digital signature of the CA over the public key to provide authenticity for the pseudonym. From Assumption~\ref{PK-assumption:pseu}, the CA provides the pseudonyms to clients with the respective signature.

\begin{remark}
\label{PK-remark:proof}
A valid token $\sigma$ for revocation list RL is a unique digital signature over the RL's seed.
\end{remark}

A token $\sigma_x^{RL_z}$ represents the unique bond of a pseudonym $p_x$ to a given revocation list $RL_z$. This bond is implemented by a digital signature performed over the unique $seed_z$ that accompanies the list, where $\sigma_x^{RL_z} = \{\textsc{h}(seed_z)\}^{K_x^-}$. By using the private key $K_x^-$ of the pseudonym (from Definition~\ref{PK-definition:pseudonym}), we enforce authentication, integrity, and non-repudiation of the token (from Assumption~\ref{PK-assumption:asymmetric}).

\begin{remark}
\label{PK-remark:list}
A RL is immutable, unique, and cannot be tempered with.
\end{remark}

All the revocation lists are generated and published by the CA. For a new $RL_z$, the CA generates a unique and random $seed_z$, then computes a revoked access token for each revoked pseudonym as $\sigma_i^{RL_z} = \{\textsc{h}(seed_z)\}^{K_{p_i}^-}$ (from Remark~\ref{PK-remark:proof}) and then encodes each revoked access token using a one-way hash function as $\textsc{h}(\sigma_i^{RL_z})$ (from Assumption~\ref{PK-assumption:function}). The encoded set of revoked access tokens \textit{sreat}$_z$ is then bound to the $seed_z$ through a digital signature that proves its authenticity, resulting in $\{ { \textsc{h}( \textit{sreat}_z \mathbin\Vert seed_z) \}}^{K^-_{CA}}$. Any atempt to tamper the RL will invalidate its signature. The CA will never reuse the same seed, making each encoded set of revoked access tokens unique and immutable.

\begin{remark}
\label{PK-remark:authenticating}
A valid authentication request is defined as:
\[ \langle x, \sigma_x^{RL_z}  \rangle \]

\end{remark}

When authenticating, the client must provide the authentication request where $x$ is defined as the public information of some pseudonym $p_x$, specifically the public key and the CA signature:

$ \langle K_{p_x}^+, \{\textsc{h}(K_{p_x}^+)\}^{K_{CA}^-} \rangle $. 

Before presenting this information towards the verifier, the client requests the current revocation list $RL_z$ and generates $\sigma_x^{RL_z}$ from Remark~\ref{PK-remark:proof}. Following Assumption~\ref{PK-assumption:revoked}, if the client does not find $\sigma_x^{RL_z}$ in the encoded set of revoked access tokens for $RL_z$, it proceeds with the authentication by sending the authentication request $\langle x, \sigma_x^{RL_z} \rangle$ to the verifier.

\begin{theorem}
\label{PK-theorem:unlinkability}
Two different authentications cannot be linked to the same client.
\end{theorem}

\begin{proof}
Consider an authentication from client $c$ that provides an authentication request $\langle x, \sigma_x^{RL_i} \rangle$ for pseudonym $x\in {P}_c$ in face of  some revocation list $RL_i$ (where ${P}_c$ is the set of pseudonyms of client $c$). Consider another authentication from client $c'$ that provides the request $\langle y, \sigma_y^{RL_j} \rangle$ for pseudonym $y\in{P}_{c'}$ in face of some revocation list $RL_j$. For an attacker to successfully link the two authentications, it needs to infer that $c=c'$. There are two ways for an attacker to achieve this goal. 

One is to assert that $x$ and $y$ belong to the same client,  i.e.:

\begin{equation}
\textsc{assert} (\exists_{c,{P}_c},~x \in {P}_c ~\land~ y \in {P}_c)
\label{eq:samePseu} 
\end{equation}

This condition holds true when $x$ and $y$ belong to the same pseudonym $p\in P_c$, and according to Remark~\ref{PK-remark:authenticating}, this would require $x$ and $y$ to share the same public key $K_{p_c}^+$. However, by Assumption~\ref{PK-assumption:unused}, the client never uses the same pseudonym twice. Additionally, according to Assumption~\ref{PK-assumption:pseu}, the CA generates all the public and private keys of each pseudonym using a secure, random, and invertible cryptographic scheme, this will result in unlinkable and random pseudonyms by construction. Therefore, an attacker will always observe $x\neq y$ and never be able to assert Equation~\ref{eq:samePseu} as true.

The other way is to use the information in some other revocation list $RL_k$ to link pseudonyms $x$ and $y$. Let us use \textit{RL.sreat} to denote the encoded set of revoked access tokens of revocation list RL.

\begin{equation}
\textsc {assert} (\exists_{RL_k},~\textsc{h}(\sigma_x^{RL_k}) \in \textit{RL$_k$.sreat}~\land~  \textsc{h}(\sigma_y^{RL_k}) \in \textit{RL$_k$.sreat})
\label{eq:rlk} 
\end{equation}

This condition holds true if the attacker can verify that exist a list $RL_k$ that contains encoded tokens for both pseudonyms $x$ and $y$. Because tokens are encoded by secure one-way functions, the attacker cannot extract $\sigma_x^{RL_k}$ and $\sigma_y^{RL_k}$ from \textit{RL$_k$} or \textit{RL$_k$.sreat} and therefore these must be provided by the client. By construction (Assumption~\ref{PK-assumption:revoked}), the client will not provide $\sigma_x^{RL_k}$  if $\textsc{h}(\sigma_x^{RL_k}) \in \textit{RL}_k\textit{.sreat}$. Thus, the attacker can only have access to $\sigma_x^{RL_k}$ if there is some other $RL_{k'}$ such that $\sigma_x^{RL_k}=\sigma_x^{RL_{k'}}$, however, by Remark~\ref{PK-remark:list}, this is impossible, because tokens are unique, given that they depend cryptographically on different seeds. Therefore, the attacker is not able use some public revocation list to assert Equation~\ref{eq:rlk} as true.

Therefore, we can conclude that the attacker is incapable of leveraging the available information to assert either Condition~\ref{eq:samePseu} or~\ref{eq:rlk} as true. This implies that the attacker cannot infer whether $c = c'$. Consequently, an attacker is unable to link two different authentication attempts to a single client.

\end{proof}

\subsection{Auditability}

We now provide a proof that \mysystem offers full auditability.

\begin{lemma}
\label{PK-lemma:inclusion}
After a pseudonym $p_x$ is revoked, any RL presented to clients must encode $\sigma_x^{RL_z}$ in the RL's encoded set of revoked tokens.
\end{lemma}

\begin{proof}
A client authenticates against a given RL by selecting a pseudonym $p_x$ and presenting $\langle x, \sigma_x^{RL_z} \rangle$. Authentication is granted if $\sigma_x^{RL_z}$ has not been added to the RL's encoded set of revoked tokens. Thus, for revocation of pseudonym $p_x$ to succeed, $\sigma_x^{RL_z}$ must be added to the RL's encoded set of revoked tokens of all RL presented to a client after revocation.
\end{proof}

\begin{lemma}
\label{PK-lemma:exclusion}
A client authenticates using pseudonym $p_x$ only if, when presented with a revocation RL, it cannot find $\sigma_x^{RL_z}$ in the RL's encoded set of revoked tokens.
\end{lemma}

\begin{proof}
A client authenticates against a given RL by selecting a pseudonym $p_x$ and presenting $\langle x, \sigma_x^{RL_z} \rangle$. By construction (see Assumption~\ref{PK-assumption:revoked}) a client will only authenticate using $p_x$ if it cannot find $\sigma_x^{RL_z}$ in the RL's encoded set of revoked tokens.
\end{proof}

\begin{theorem}
\label{PK-theorem:auditability}
\mysystem ensures auditability
\end{theorem}

\begin{proof}
The proof, is by contradiction.  A client is not guaranteed auditability if it attempts to authenticate with a pseudonym $p_x$ that has been revoked. Let $p_x$ be a pseudonym used by some client to perform authentication against some revocation list RL. By Lemma~\ref{PK-lemma:inclusion}, if pseudonym $p_x$ has been revoked $\sigma_x^{RL_z}$ must be encoded in the RL's encoded set of revoked tokens. By Remark~\ref{PK-remark:list} RLs are immutable and cannot be tempered with by the verifier. By Lemma~\ref{PK-lemma:exclusion}, if the client authenticates,  it cannot find $\sigma_x^{RL_z}$ in the RL's encoded set of revoked tokens. A contradiction.
\end{proof}

\section{Implementation}
\label{sec:reduce}

\mysystem requires revocation tokens to be encoded using a one-way function, in order to create a revocation list that can be safely distributed to verifiers. Such one-way function must be efficient and produce compact results. We show how revocation lists can be implemented using Bloom Filters~\cite{bloom1970space}. We propose three alternative implementations of the RL that aim to reduce the amount of information that clients and verifiers exchange during each authentication process, while maintaining all the security properties inherent to the Revocation List, namely its integrity and authenticity:   1) an implementation based on a single Bloom Filter; 2) an implementation based on Hierarchical Bloom filter Arrays; and 3) an implementation based on Redactable Signatures. Another potential source of inefficiency is the computation of a new revocation list, that can be slow. This can be circumvented by a pre-computation strategy that we also describe in this section.  

\subsection{RLs based on a single Bloom Filter}

\begin{figure}[h!t]
    \centering
    \includegraphics[width=\columnwidth]{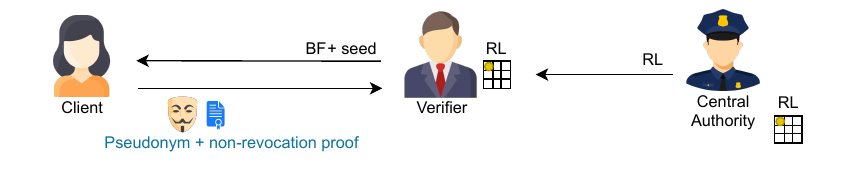}
    \caption[Authentication with a single  Bloom Filter Arrays.]{Authentication with a single BFs.}
    \label{fig:singleBF}
\end{figure}

Bloom Filters~\cite{bloom1970space} are a probabilistic structure designed to store members of a set and to determine if a given element is member. In pseudonym systems, verifiers often use Bloom Filters for the storage of the revoked pseudonyms. Bloom Filters are very efficient and have constant computation costs \textit{O}(1) for storing and searching elements. The filter is made up of \textit{N} addressable bits, with addresses from 0 to \textit{N}-1, being \textit{N} its size. Besides its size, a filter also has another parameter which is the number of hash functions, \textit{k}. When an element is to be inserted in the filter, it is hash coded using the hash functions, obtaining \textit{k} hash codes, that will be used as addresses inside the filter. After that, the bits of the filter with that addresses are set to 1. To test whether an element belongs to the filter, the element is hash coded using the hash functions obtaining \textit{k} addresses, the element is considered to be in the filter if and only if the bits in the filter with that addresses are all set to 1. When searching for an element, Bloom Filters may indicate that the element is in the filter when it is not, allowing false positives to happen. On the contrary, false negatives are not possible. The rate of false positives depends on the size of the filter, \textit{N}, the number of elements inserted, \textit{m}, and the number of hash functions used, \textit{k}. The rate of false positives can be calculated using the following formula:

\begin{equation}
P = (1-(1-\frac{1}{m})^{kn})^k
\end{equation}

\fakedescription{Pseudonym Revocation}
The certification authority (CA) creates a single Bloom filter (BF) and encodes all revocation tokens into the BF. Finally, the BF is signed by the CA and sent to the verifiers. The revocation list, RL, is then composed of, the bloom filter, BF, the seed and the signature.

\[
\textit{RL} = \langle 
\textit{BF}, 
\{\textit{digest}(\textit{BF, seed})\}^{K^-_{\textit{CA}}}, 
\textit{seed} 
\rangle
\]

\fakedescription{Access Control}
During the authentication process the client downloads the Bloom filter. The authenticity of the filter is verified by the client using the certification authority's signature. Then, the client chooses a pseudonym $p$ and checks if this pseudonym is not marked as revoked in the received filter. If the client concludes that it has not been revoked, it can then complete the authentication by sending the pseudonym $p$ and the corresponding proof to the verifier.

\subsection{RLs based on HBFAs}

Hierarchical Bloom filter Arrays \cite{hba} are a data structure based on Bloom Filters that aim to increase the efficiency when large Bloom Filters are used. These filters use a hierarchical structure composed of multiple Bloom Filters of different sizes (with the filters in the higher positions of the hierarchy being smaller than the filters in the lower positions), where the insertion/test of presence is done by utilizing multiple filters. Typically, a membership test in the data structure requires performing sequential presence tests across the various filters until a conclusion is reached. The Revocation List can be implemented using multiple Bloom Filters of different sizes, each containing exactly the same elements. These filters are organized sequentially and in ascending order of their size.

\fakedescription{Pseudonym Revocation}
During the revocation of a client, the certification authority creates a set of $n$ Bloom Filters. Then, all revocation tokens inserted into all Bloom filters. Finally, the central authority signs each of these filters with its private key, protecting the authenticity and integrity of each filter.
Subsequently, these filters, their respective digital signatures, and the generated seed are aggregated to form the revocation list which is sent to the verifiers: 

\[
\textit{RL} = \left\langle 
\left[ \textit{BF}_1, \ldots, \textit{BF}_n \right], 
\left[ 
\{\textit{digest}(\textit{BF}_1, \textit{seed})\}^{K^-_{\textit{CA}}}, \ldots, 
\right.
\right.
\]
\[
\left. 
\{\textit{digest}(\textit{BF}_n, \textit{seed})\}^{K^-_{\textit{CA}}} 
\right], 
\textit{seed} 
\rangle
\]

\fakedescription{Access Control}
During the authentication process, represented in Figure~\ref{fig:hfilters}, the client starts by downloading the first BF from the Revocation List, which is the one with the smallest size, as well as the random seed, \textit{seed}, associated with the Revocation List. Then, it chooses a pseudonym $p$ and checks if this pseudonym is marked as revoked in the previously received filter. Since the false positive rate of a BF varies with its size, and the client begins by receiving the smallest filter, it is not unlikely that the client encounters a false positive at this step. In that case, the client will then download the remaining filters from the Revocation List in increasing order of size until it reaches a filter where the pseudonym is not marked as revoked. If it reaches the last filter and the pseudonym is marked as revoked in that filter, the client assumes that it has indeed been revoked and cancels the authentication. The authenticity of each filter is verified by the client using the certification authority's signature. If the client concludes that it has not been revoked, it can then complete the authentication by sending the pseudonym $p$ and the corresponding proof to the verifier.

\begin{figure}[h!t]
    \centering
    \includegraphics[width=\columnwidth]{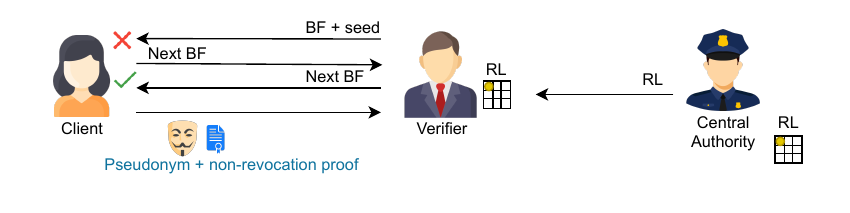}
    \caption[Authentication with Hierarchical Bloom Filter Arrays.]{Authentication with HBFAs.}
    \label{fig:hfilters}
\end{figure}

HBFA are advantageous in cases where the client can confirm that the pseudonym has not been revoked by checking the first few filters (ideally, relying only on the smallest filter in most cases); however, it may worsen the original solution if the client has to download several filters. In most cases, the client would be able to authenticate using the first pseudonyms by checking only the smaller filters, making this an efficient solution overall.

\subsection{RLs based on Redactable Signatures}

Redactable Signatures~\cite{johnson2002homomorphic} enable the sending of a digitally signed message, with the ability to delete certain parts of it, while still allowing the recipient to verify the security properties of the message, such as integrity and authenticity. The process of creating an editable signature begins by dividing the message into removable parts. A Merkle Tree is constructed by associating each part of the message with a leaf node of the tree, calculating its cryptographic hash, and then recursively building the tree up to the root node, where each node's value is the hash of the concatenation of its child nodes' hashes. After that, the sender digitally signs the root node.

When sending the message, the sender deletes the parts it wish to omit and sends the remaining parts to the recipient. Along with the message, it also send the hashes of the leaf nodes corresponding to the deleted parts. These hashes can be compressed using the tree structure by sending the hashes of internal nodes. The recipient uses the received information to reconstruct the hash of the root node and verify the sender's digital signature on that node. The  signatures also incorporate a random component during the signature creation process, which prevents the recipient from recovering the deleted parts through brute force attacks.

We use these ideas to derive an implementation of the RL where revoked pseudonyms are first inserted in a Bloom filter and then the filter is encoded as a redactable signature, such that only the relevant parts of the Bloom filter need to be sent to the client.

\textit{Pseudonym Revocation:} In this implementation, the revocation list is first encoded as a Bloom filter, which is then divided into multiple segments of configurable size. The cryptographic hash of each segment is associated with the child nodes of a Merkle tree. The remaining nodes of the tree are generated recursively by hashing the concatenation of the hashes of the child nodes. Finally, the central authority generates a digital signature $\sigma$ by concatenating the root hash, \textit{root\_hash}(\textit{BF}), with the random seed associated with the Revocation List, as illustrated in Figure~\ref{fig:merkleTree}. The generated elements form the RL which is then distributed to the verifiers:

\[
\textit{RL} = \langle 
\textit{BF}, 
\{ \textit{digest}(\textit{seed}, \textit{root\_hash}(\textit{BF})) \}^{K^-_{\textit{CA}}}, 
\textit{seed} 
\rangle
\]

\begin{figure}[h!t]
    \centering
    \includegraphics[width=0.8\columnwidth]{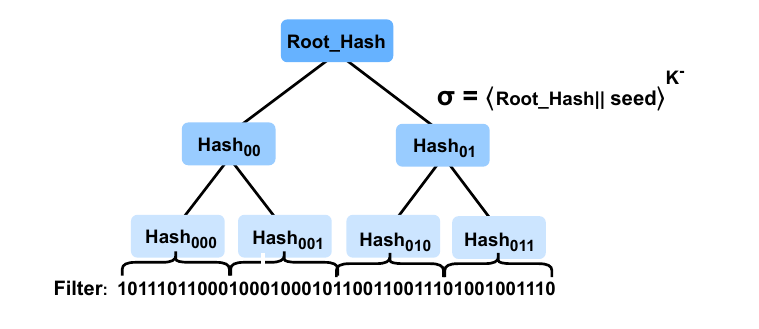} 
    \caption{Redactable Signature Creation.}
    \label{fig:merkleTree}
\end{figure}

\fakedescription{Access Control}
In the access control phase, shown in Figure \ref{fig:rSignatures}, the verifier begins by sending the random seed to the client.
With this information, the client selects a pseudonym $p$ that it has not used yet and with that pseudonym $p$, computes the non-revocation proof, 
calculates the exact positions of the BF that needs to verify to attest their revocation status,  and requests these positions.  The verifier checks if the selected bits are set to ``1''. If all the bits are set, this means that the client has been revoked. In this case, the verifier simply notifies the client that it should not proceed with the authentication. If there is at least one bit that is not set, the client has not been revoked. In this case the verifier selects one of the segments that have a bit  set to ``0''. It sends this segment along with the necessary hashes for the client  to reconstruct the path from that segments to the root node. The client can then check that segment has not been tempered by the verifier, given that any alteration would necessarily affect the root hash. If the client concludes that it has not been revoked, it completes the authentication process by sending their pseudonym $p$ and the corresponding proof to the verifier.

\begin{figure}[h!t]
    \centering
    \includegraphics[width=\columnwidth]{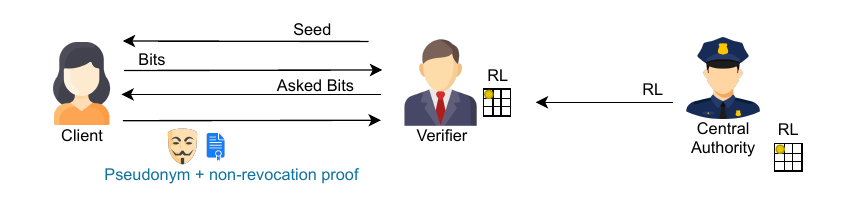} 
    \caption[Authentication with Redactable Signatures.]{Authentication with Redactable Signatures.}
    \label{fig:rSignatures}
\end{figure}

\subsection{Precomputing the RL} 

Since the most resource-intensive part of our system is the construction of the RL from scratch every time a new client is revoked, we have implemented an optimization to accelerate this process. This optimization hinges on the CA consistently maintaining a precomputed RL. In other words, the CA consistently prepares in advance a new version of the RL, selecting a random seed and generating all the revoked access tokens associated with that RL for all previously revoked users. When there is a need to revoke a new client, it is merely a matter of generating the revoked access tokens for the new client and inserting them into the encoded set of revoked access tokens. This precomputation significantly reduces the time required to generate a new RL and to publish it, making revocation more effective.

\section{Evaluation}
\label{secshort:eval}

We evaluate \mysystem in two orthogonal dimensions. Firstly, we assess the time it takes to construct a revocation list. Then, we assess the efficiency of authentication when using the different implementations presented in Section~\ref{sec:reduce}. For our evaluation, we used an Intel NUC10i7FNB, with an Intel i7-10710U CPU with Intel SGX, 16GB RAM, and Ubuntu 20.04 LTS, a setting similar to other papers that also target edge networks~\cite{paper_CCS}.

The  code used for the evaluation is open source and available in the following anonymous Git repository: \url{https://github.com/LaraAuth12/LARA}.

\subsection{Time to Generate a RL}

\subsubsection{Single Filter}

We compare the time required to create a revocation list with \mysystem against the two main solutions that offer backward unlinkability in public key encryption, namely Haas et al.~\cite{haas2011efficient} and RRPs~\cite{paper_CCS}, despite the fact that these systems do not provide revocation auditability, a distinctive feature of our system. Additionally, to ensure a fair comparison, we selected the ideal parameters for RRPs, an epoch of one month and a time slot of one minute. This means that for Haas et al. it requires at least one different pseudonym for each minute of the month.

\begin{figure}[h]
  \centering
    \centerline{\includegraphics[width=\columnwidth]{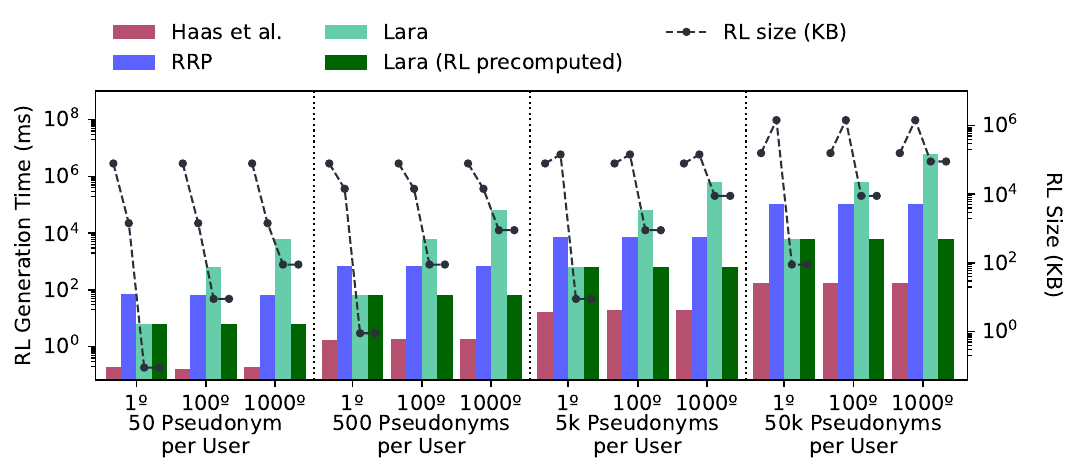}}
    %usecase.png}}
  \caption{RL Generation Time.}
  \label{fig:eval}
\end{figure}

We compare the the time required to generate a new version of the RL and the size this list can reach. Figure~\ref{fig:eval} shows the latency required to create a new version of the RL for the first revoked user, the one-hundredth, and the one-thousandth, while varying the number of pseudonyms each client owns. This figure depicts value obtained with the  implementation that uses a single Bloom filter. It is important to note that with each new revocation, we need to recreate the RL from scratch and generate the necessary information for all previously revoked clients. Therefore, the more clients have been revoked, the longer the latency to create an RL will be. Using the strategy of precomputing the RL, \mysystem latency similar to Haas et al., with a worst-case time of just under 10 seconds to generate. Another visible advantage of our system is that the size of the RL is much more efficient, maintaining the same flexible properties as RRPs. In our scheme, clients store only the desired number of pseudonyms, and the revocation only requires a single token per pseudonym, while RRPs need $\log(\frac{epoch}{slot})$ of revocation data in the RL per pseudonym. It's worth noting that although Haas et al. is more efficient, it offers fewer security properties and flexibility, as it forces clients to carry many pseudonyms, even if they do not need them.

\subsubsection{HBFA}

\begin{figure}
  \centering
  \includegraphics[width=\columnwidth]
  {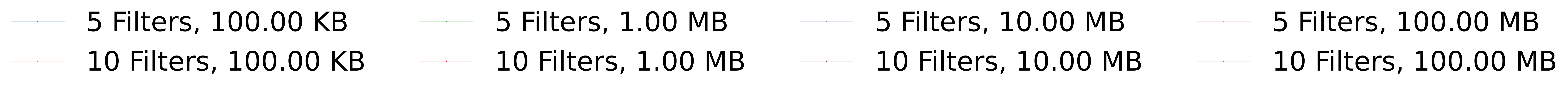}
  \centering
  \begin{subfigure}[b]{\columnwidth}
    \includegraphics[width=\columnwidth]{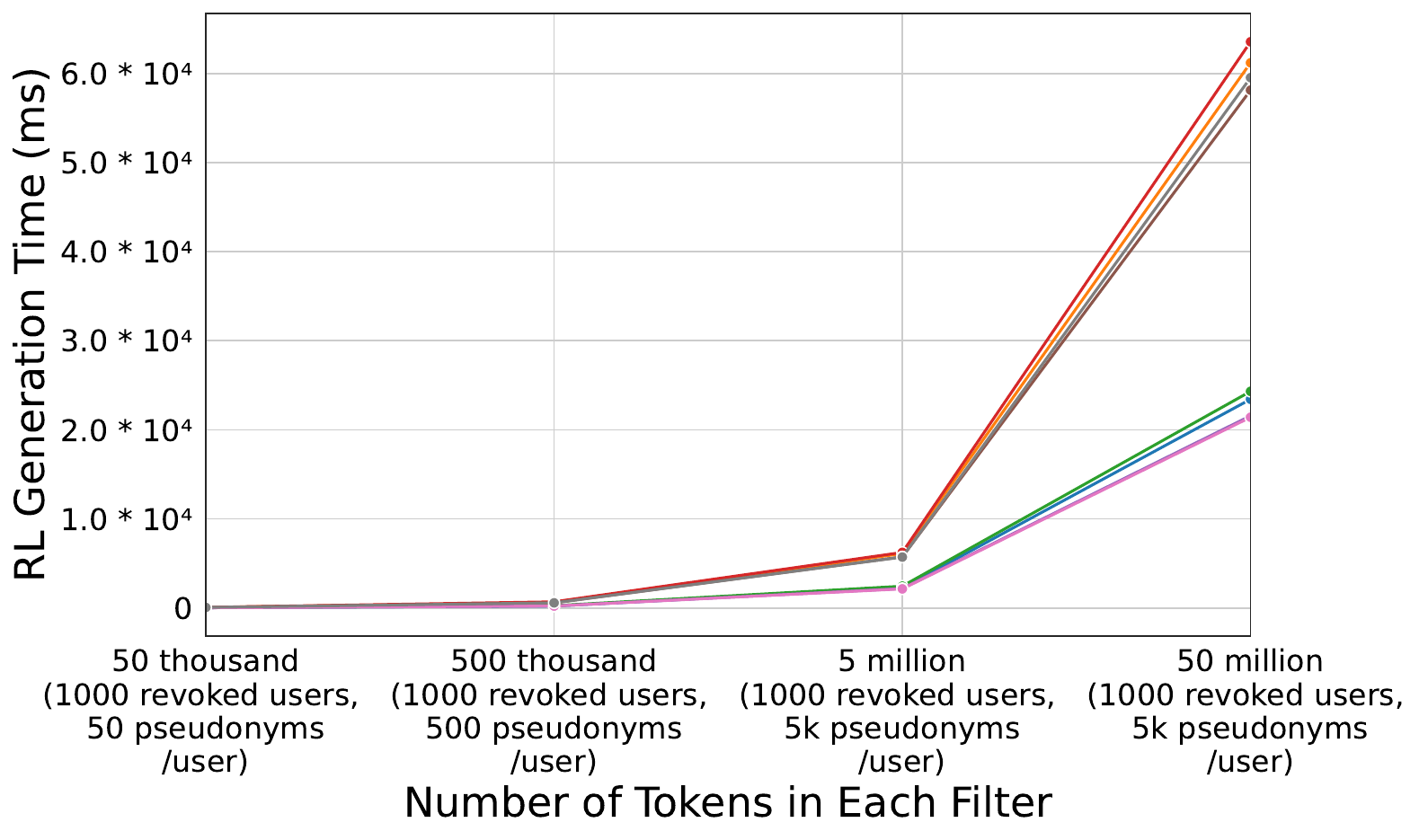}
    \caption{5 hash functions}
    \label{fig:HBFA_Latencya}
  \end{subfigure}
  
  \begin{subfigure}[b]{\columnwidth}
  \includegraphics[width=\columnwidth]{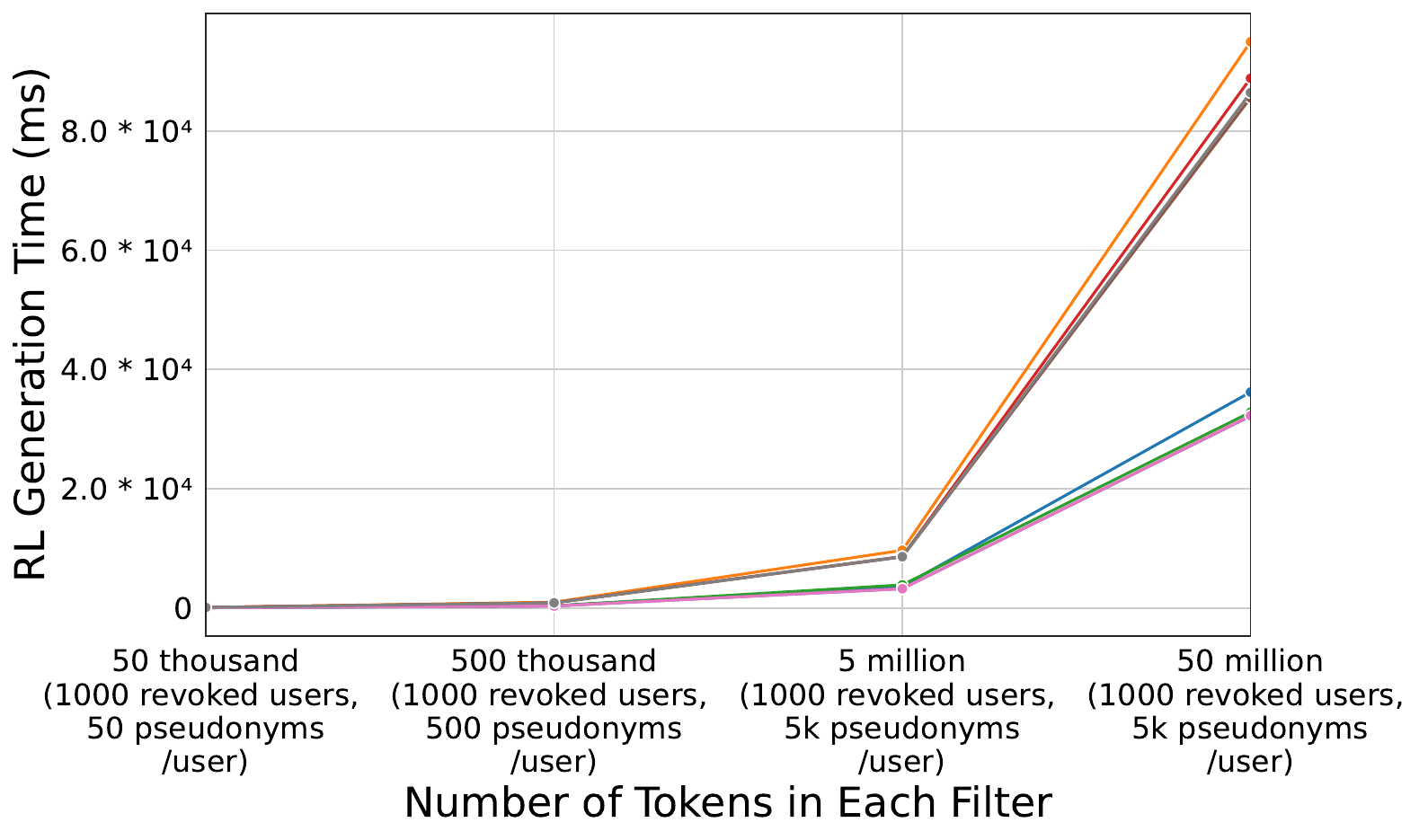}
    \caption{10 hash functions}
    \label{fig:HBFA_Latencyb}
  \end{subfigure}
  \caption{HBFA's Overhead}
  \label{fig:HBFA_Latency}
\end{figure}

Figure~\ref{fig:HBFA_Latency} illustrates the results of a series of tests we conducted to assess the additional overhead involved in generating the revocation list using the Hierarchical Bloom Filter Arrays technique, as well as its impact on the latency required to generate a revocation with this approach.

There are several factors that may impact the necessary additional overhead in generating a revocation list with multiple Bloom filters, instead of just one. We have conducted a series of tests to measure the additional latency this would take, varying parameters such as: the number of filters that our revocation list was composed of, the number of hash functions of each filter, the size of those filters and the number of tokens to be inserted in each filter.

By analyzing our results it is possible to notice that the number of hashes, the number of tokens to be inserted and the number of filters, have a linear impact in the additional overheard, doubling the latency when doubling the number of hashes or the number of tokens, due to the number of bits that are necessary to set to "1" also double. We can also conclude that the size of the Bloom filters do not have a significant impact on this metric, due to the nature of Bloom filters being highly efficient and constant with respect to their operations.

Also, we can conclude that, depending on the system scale, the additional overhead can be negligible, being only a few milliseconds for most cases until 5 million tokens, starting to impose a serious overhead when we hit the scale of 50 million tokens, reaching the magnitude of 60 seconds or more.

\subsubsection{RS}

\begin{figure}[h]
  \centering
    \centerline{\includegraphics[width=\columnwidth]{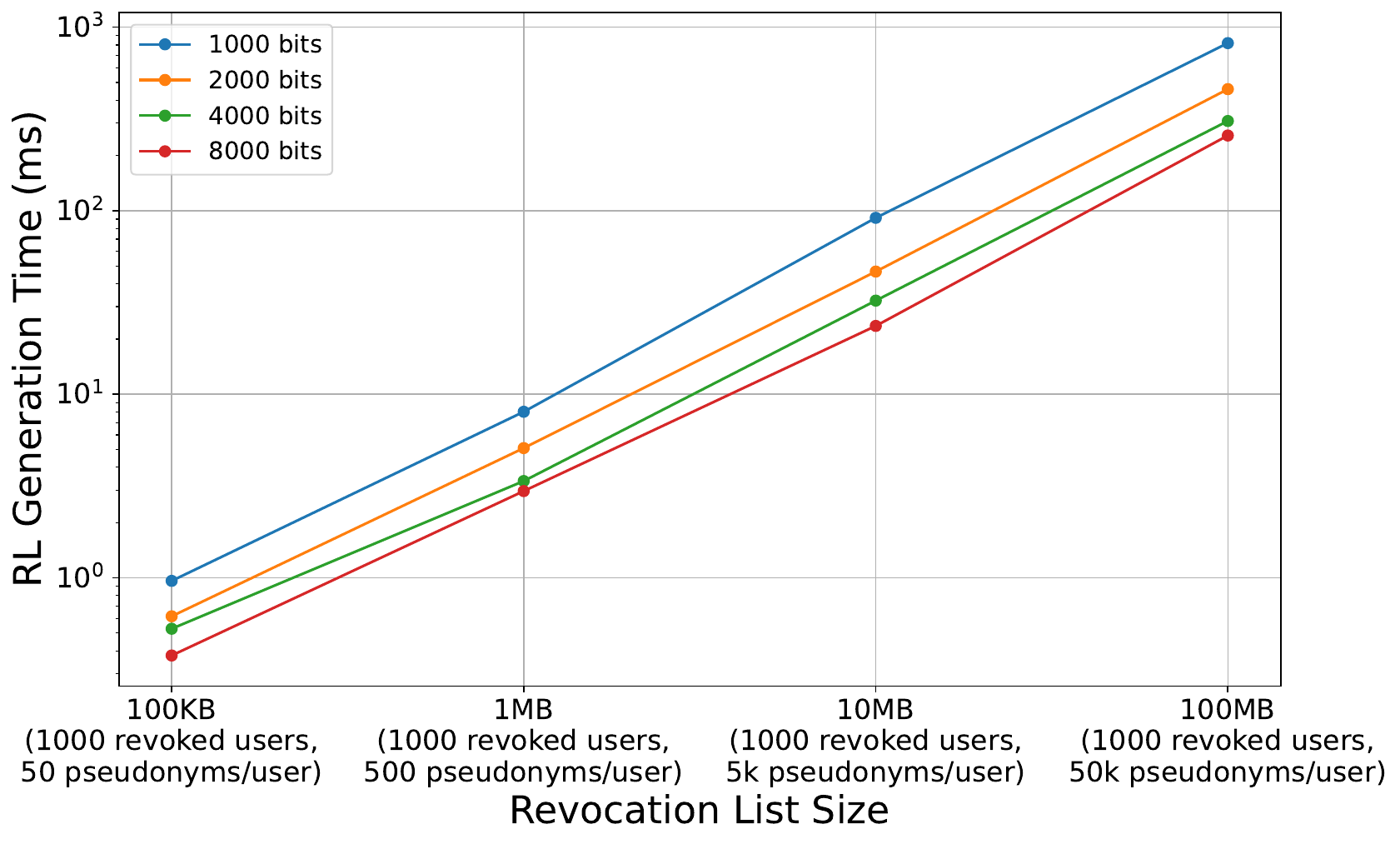}}
    %usecase.png}}
  \caption{Redactable Signature Overhead.}
  \label{fig:redactSigTime}
\end{figure}

To quantify the overhead introduced by our Redactable Signature technique, we conducted experiments measuring the time required to generate a signature over a revocation list, varying both the size of the Bloom filter and the size of the chunks into which it was divided.

As shown in Figure~\ref{fig:redactSigTime}, the signature generation time increases linearly with the Bloom filter size, as expected. Additionally, using smaller chunks results in longer generation times due to the larger number of leaf nodes in the Merkle tree, which increases the number of required hash computations. In contrast, larger chunks yield smaller trees and thus lower computational overhead. The latency required to generate a redactable signature ranges from just 1 millisecond or less for Bloom filters of 100KB to hundreds of milliseconds for filters as large as 100MB.

This highlights a trade-off between redaction granularity and efficiency. As discussed later (see Figure~\ref{fig:varParaAss2}), larger chunks reduce signature generation time but require more data transfer during authentication. This trade-off becomes particularly relevant for Bloom filters exceeding 100MB. For smaller filters, generation time remains below one second across all chunk sizes, making the overhead negligible in practice.

\subsubsection{Comparison}

Based on the previous, we can assess the impact of our proposed techniques on revocation list generation time. The RS-based technique introduces minimal overhead—adding no more than a single extra second for Bloom filters smaller than 100MB. In contrast, the overhead introduced by the HBFA-based technique is significantly influenced by the number of tokens added to the revocation list. In the worst-case scenario analyzed (1,000 revoked users with 50,000 pseudonyms each), the baseline generation time using a single Bloom filter exceeds 1,000 seconds. The RS-based method increases this by less than one second—an overhead of under 0.1\%. Meanwhile, the HBFA-based approach adds approximately 60 seconds for Bloom filters using 5 hash functions, corresponding to an overhead of about 6\%.

\subsection{Authentication Performance}

\subsubsection{Using a single BF}

\begin{figure}[h]
  \centering
    \centerline{\includegraphics[width=\columnwidth]{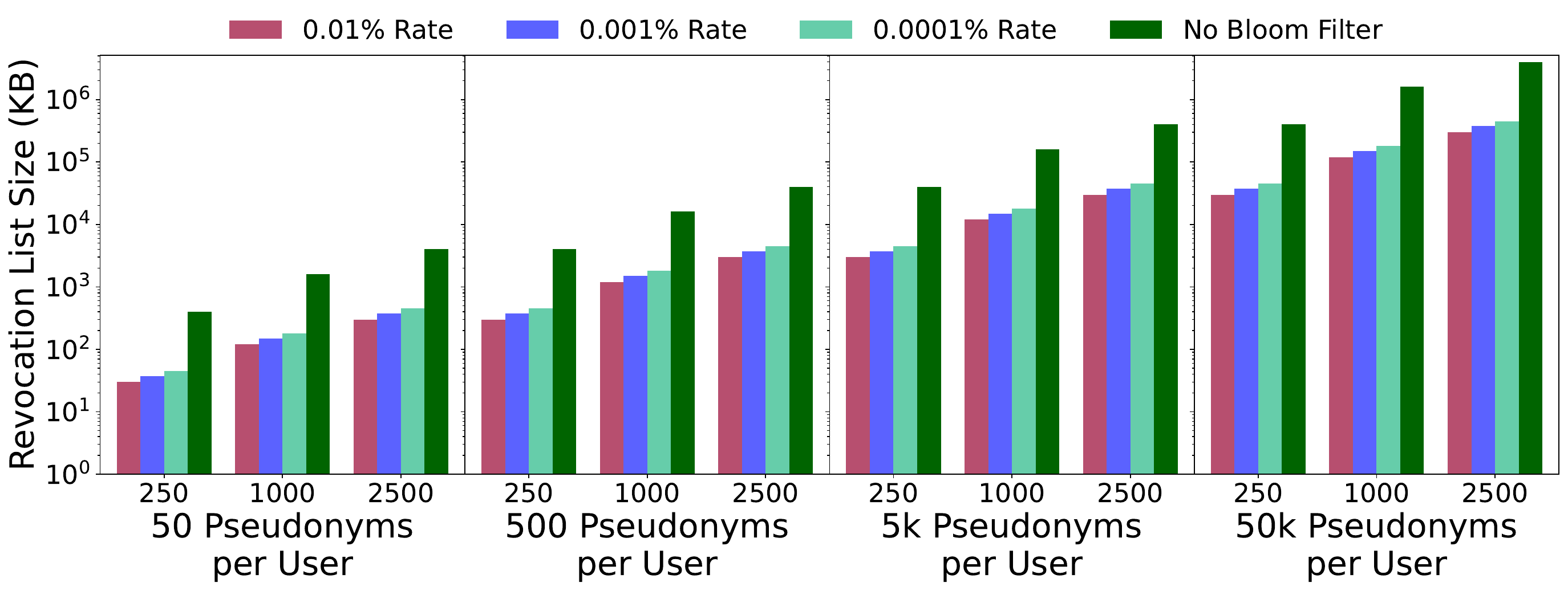}}
    %usecase.png}}
  \caption{RL Size Growth.}
  \label{fig:RLSizeGrowth}
\end{figure}

To evaluate authentication using a single Bloom Filter (BF), we analyzed the size of revocation lists while varying the system scale and the Bloom Filter properties used to implement the revocation list. The size of a Bloom Filter, given a fixed number of inserted elements, is primarily influenced by two factors: the false positive rate and the number of hash functions. To minimize the Bloom Filter size, we calculated the optimal number of hash functions for each scenario. Figure~\ref{fig:RLSizeGrowth} presents the results of our tests across different scenarios, comparing various false positive rates, and evaluating the impact of using a Bloom Filter versus not using one.

Our study demonstrates that Bloom filters significantly reduce the size of revocation lists. We tested Bloom filters with false positive rates of 0.01\%, 0.001\%, and 0.0001\%, showing their effectiveness in minimizing the size of the revocation list. In systems like these, revocation lists grow linearly depending on the number of pseudonyms they contain. In the worst-case scenario we tested—2500 revoked users and 50,000 pseudonyms per user, simulating highly dynamic environments—the use of Bloom filters reduced the size of the revocation list by up to 90\%. In this case, maintaining a revocation list without a Bloom filter required 4GB, whereas using a Bloom filter reduced this to approximately 0.4GB.

As expected, revocation list's size decreases as the allowed false positive rate increases. However, a higher false positive rate can impact authentication performance for non-revoked users. False positives can increase authentication latency and cause certain pseudonyms assigned to legitimate users to become unusable. Therefore, selecting an appropriate trade-off between storage efficiency and authentication performance is essential.

\subsubsection{Using a HBFA}

HBFA uses multiple filters of different sizes, with several parameters that directly influence the amount of information to be transferred during authentication. Thus, we evaluated the following parameters: 1) reduction factor between each filter, starting from the largest/original filter down to the smallest one. 2) the number of filters with different sizes used in the hierarchy, i.e., how many times we reduce the largest filter by the chosen factor. 3) the false positive rate that we accept in the largest filter of the hierarchy. In Figure~\ref{fig:varParaHier}, we vary these parameters and calculate the expected amount of information related to Bloom Filters that needs to be transferred during authentication for each of these configurations, keeping the size of the largest filter fixed at 0.5GB.

\begin{figure}[t]
  \centering
  \includegraphics[width=\columnwidth]{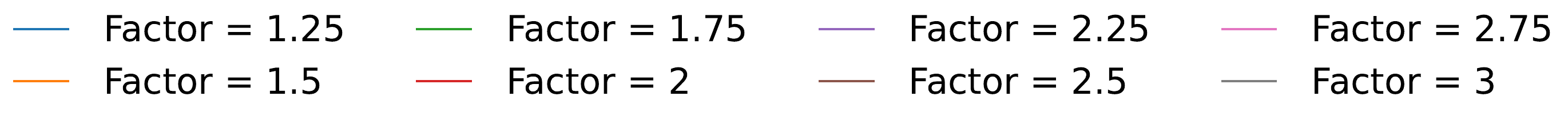}
  
  \ 
  \centering
  
  \begin{subfigure}{0.65\columnwidth}
    \includegraphics[width=\columnwidth]{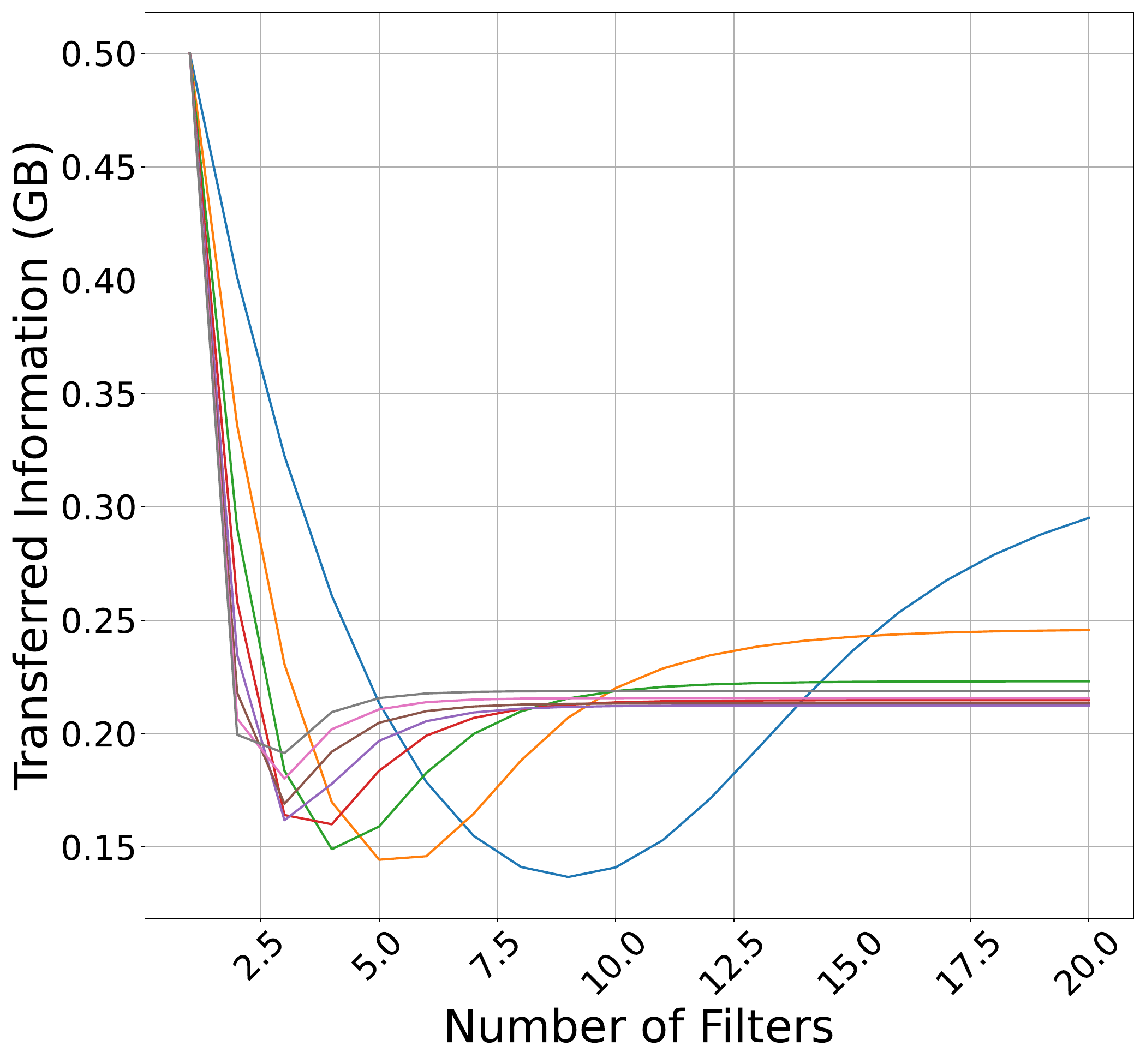}
    \caption{$0.1\%$}
    \label{fig:0.001}
  \end{subfigure}

  \begin{subfigure}{0.65\columnwidth}
    \includegraphics[width=\columnwidth]{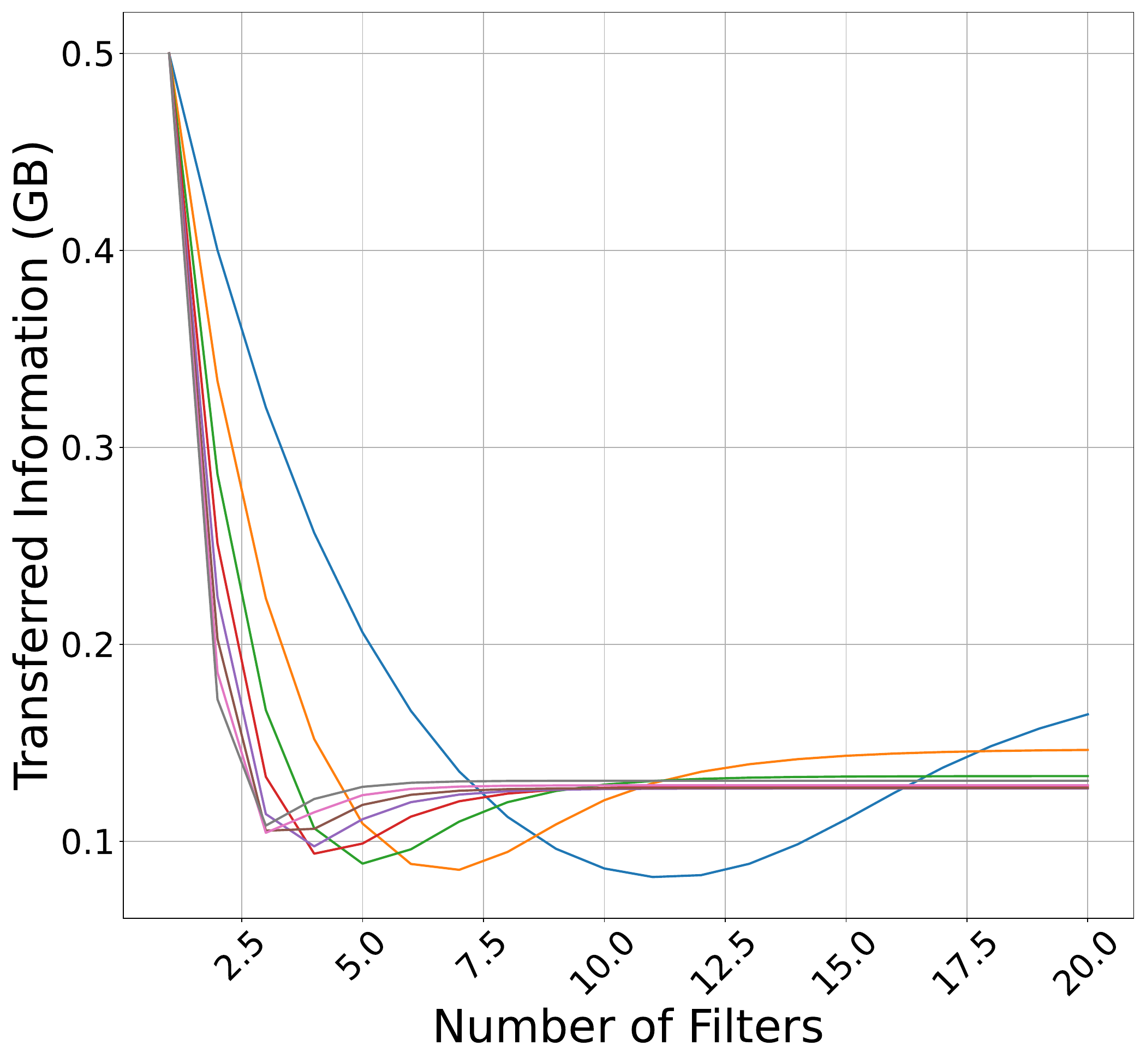}
    \caption{$0.01\%$}
    \label{fig:0.0001}
  \end{subfigure}

  \begin{subfigure}{0.65\columnwidth}
    \includegraphics[width=\columnwidth]{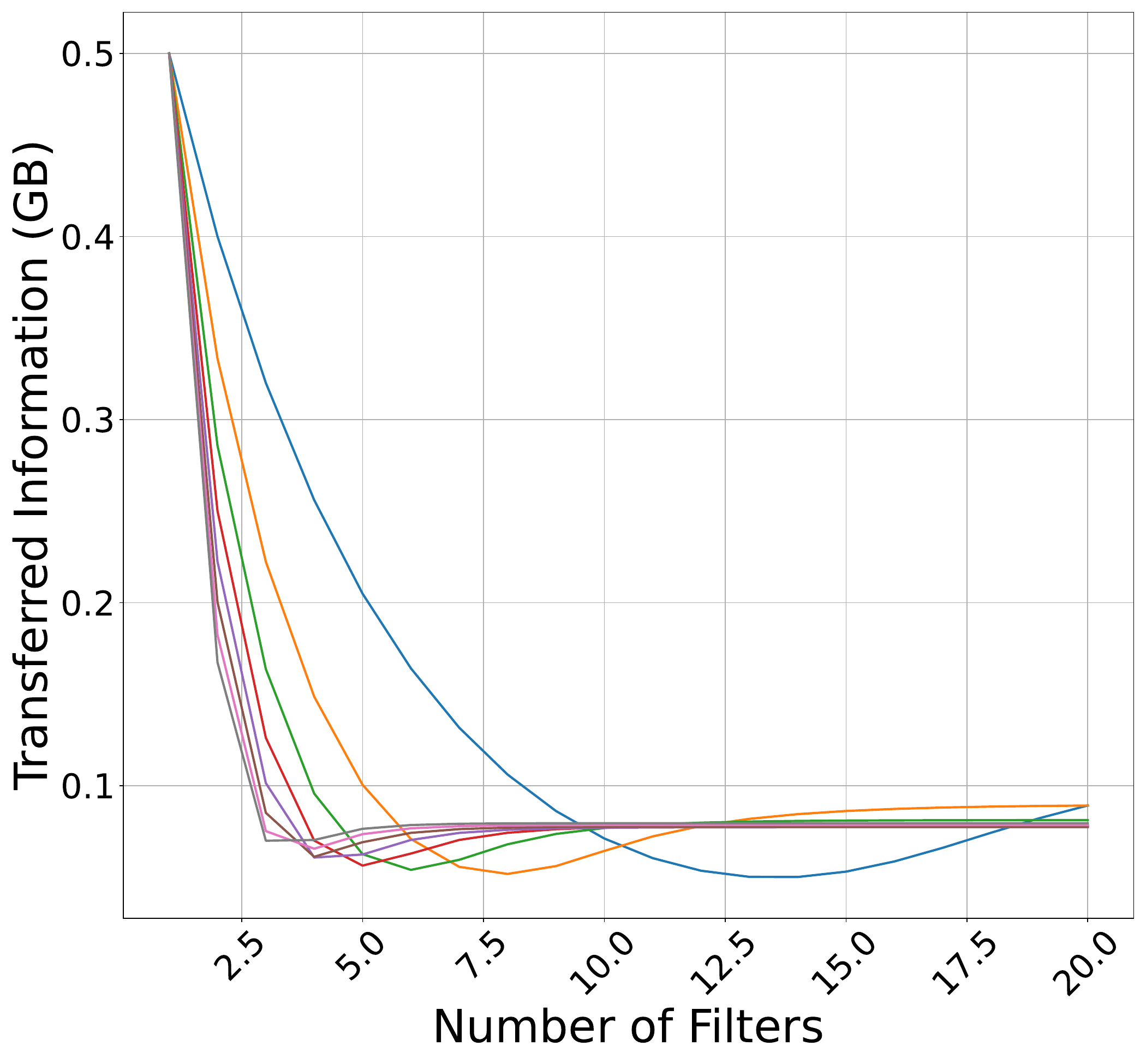}
    \caption{$0.001\%$}
    \label{fig:0.00001}
  \end{subfigure}
   \caption{Different Target False Positive Rates}
   \label{fig:varParaHier}
\end{figure}

It is worth noting that calculating the expected amount of transferred information is not trivial, as it depends on the number of filters the client needs to transfer before authenticating. To evaluate each set of parameters, we calculated the amount of information transferred for each scenario (in each scenario, the client needs to download a different number of filters before completing the audit process), and we calculated the expected value of the amount of information transferred by the client using the false positive rates of each BF to compute the probability of each scenario.

It can be observed that, when there is only one filter, the information required to transfer is always 0.5GB, representing the original filter. However, when we create more filters, the average amount of information transmitted quickly decreases to less than half with more than 5 filters. Furthermore, we observed that the smaller the reduction factor, the lower the expected amount of information transmitted during each authentication. However, it is necessary to increase the number of filters used to reach the minimum amount of information transmitted.

We also observed that the transferred information decreases when we reduce the false positive rate. This is due to the fact that reducing the rate also affects the smaller filters, as all filters will have lower false positive rates and a greater number of clients will authenticate using the first filters. Note that, in order to reduce the false positive rate while keeping the filter size fixed, it is necessary to reduce the number of items in the filter.

A good configuration for a false positive rate of 0.001\% would be to choose a factor of 2 with 4 filters, as this is a point that minimizes the information to be transmitted over the network while maintaining a reasonable number of filters. Our technique drastically reduces the amount of information transferred, in particular, transferring only about 10\% of the original amount of information on average.

\subsubsection{Using Redactable Signatures}

\begin{figure}
  \centering
  \includegraphics[width=0.9\columnwidth]
  {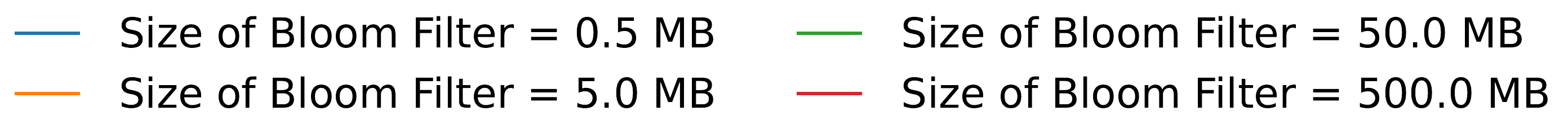}
  
  \centering
  \begin{subfigure}[b]{0.7\columnwidth}
    \includegraphics[width=0.9\columnwidth]{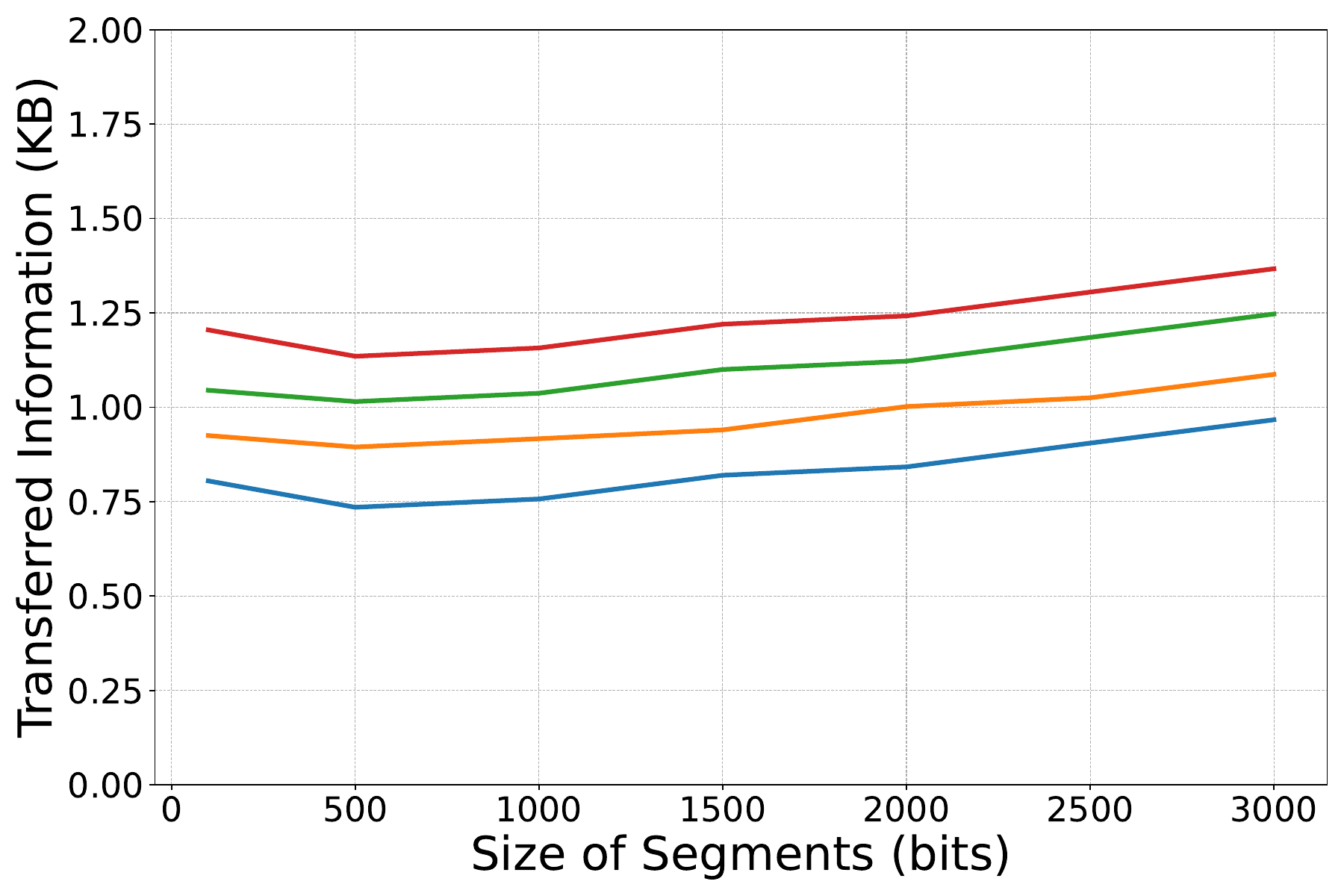}
    \caption{5 hash functions}
    \label{fig:5hashes}
  \end{subfigure}

  \begin{subfigure}[b]{0.7\columnwidth}
    \includegraphics[width=0.9\columnwidth]{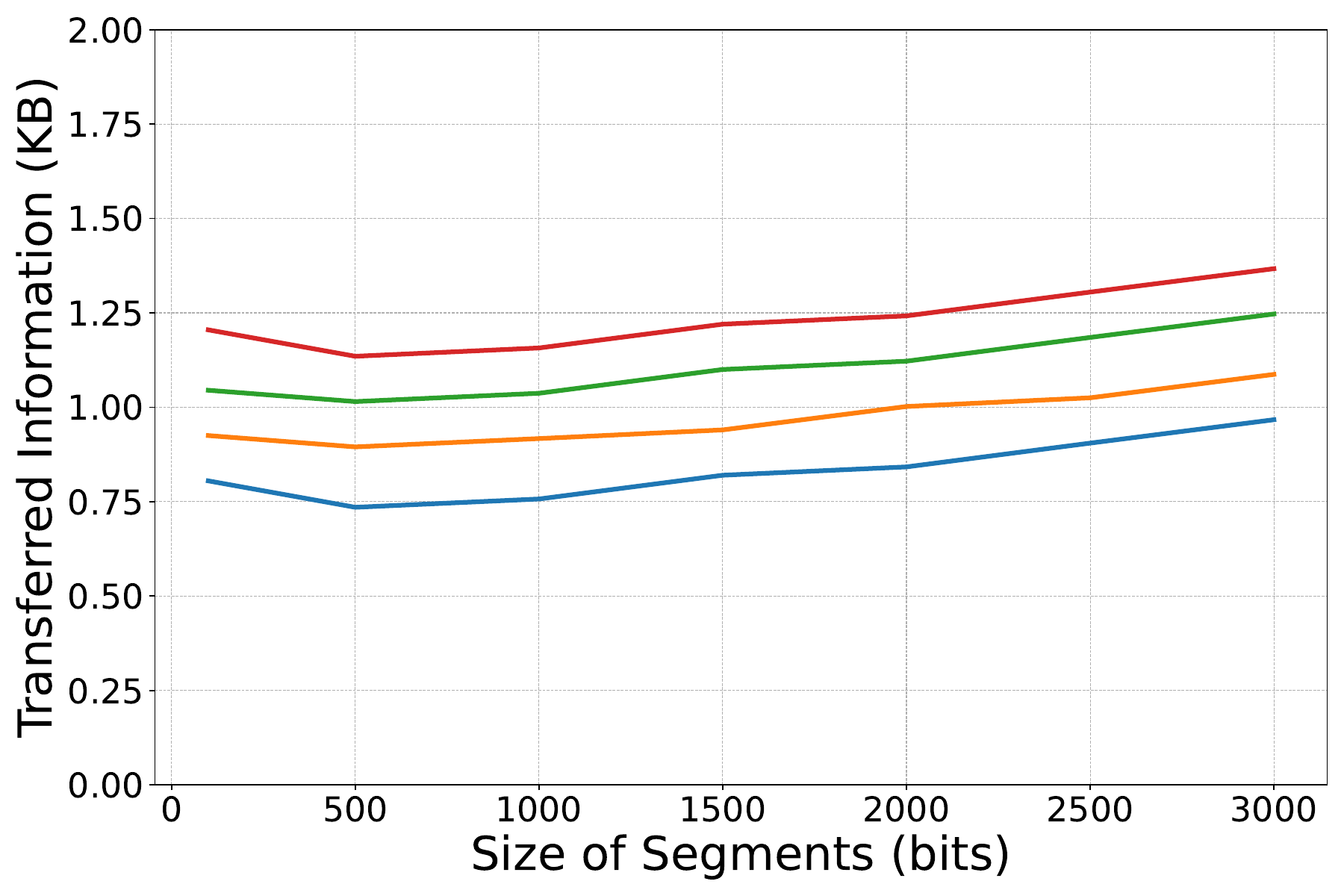}
    \caption{10 hash functions}
    \label{fig:10hashes}
  \end{subfigure}
  \caption{Amount of data exchanged}
  \label{fig:varParaAss2}
\end{figure}

Redactable Signatures use two main structures: a Bloom filter and a Merkle Tree, which is constructed from the filter. The information transferred during authentication consists mainly of segments of the BF and various hashes from the Merkle Tree.

Several parameters of the solution impact the amount of information transferred during authentication. These parameters include: 1) the number of hash functions, which directly influences the number of filter segments sent to the client, and 2) the size of the segments into which the filter is divided, since the segment size determines the number of filter segments that will be associated with the Merkle tree's leaf nodes, affecting the size of the Merkle tree and the number of hashes sent to the client.

In Figure~\ref{fig:varParaAss2}, we varied these parameters and measured the average amount of information that the verifier sends to the client during authentication. Given that only a segment is transferred to the client, the number of hash functions used does not increase the amount of information transferred during the authentication. Thus, it is possible to increase the number of hash functions of the BF, thus, reducing the false positive rate, without increasing the amount of information transferred. We can observe that the optimal configurationis  to divide the BF into segments of 500 binary digits. With this configuration,  the amount of information sent to the client is just 125MB.

\subsubsection{Comparison}

In this section, we compare the authentication performance of the three proposed \mysystem implementations. To ensure a fair comparison, we chose configurations that minimize the amount of information exchanged during authentication between verifiers and clients, based on the analysis presented earlier.

For this evaluation, we developed a set of tests in which we measured the average authentication latency of a client, varying the size of the BF that implements the Revocation List. In the case of the Hierarchical Bloom filter Arrays technique, this size corresponds to the size of the last and largest filter. For the evaluation of the technique with Hierarchical Bloom filter Arrays, we used the following configuration: a reduction factor of 2, which means that each filter is half the size of the next filter, with the number of filters in the set fixed at 4, and 5 hash functions in each filter. We considered a false positive rate of 0.01\% for the last BF, which we then used to calculate the false positive rates for the other filters. For the evaluation of the Redactable Signatures-based technique, we fixed the segment size at 500 binary digits and the number of hash functions at 5, measuring the authentication latency.

\begin{figure}[t]
    \centering
    \includegraphics[width=0.7\columnwidth]{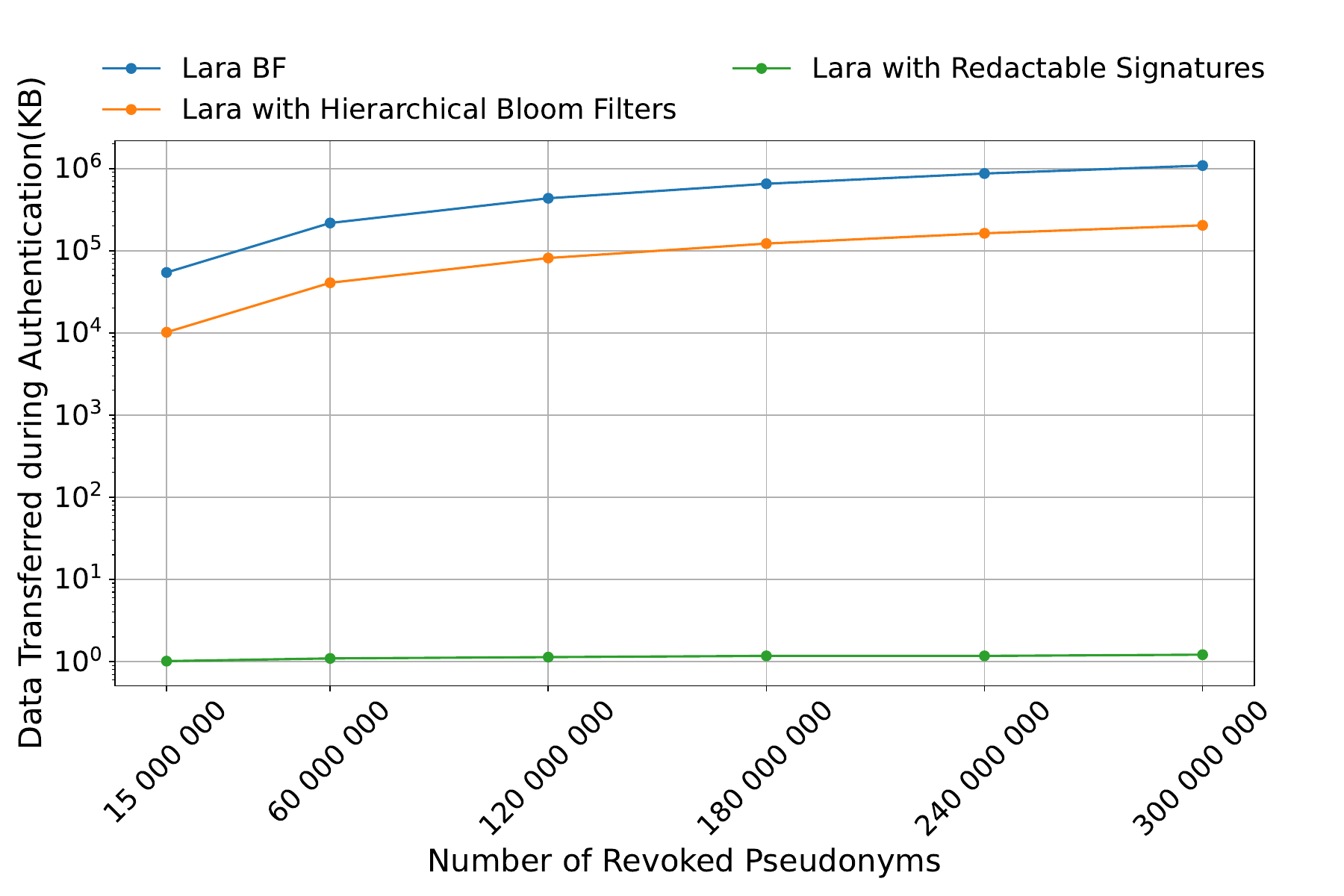}
    \caption{Data Transferred vs \# Revoked Pseudonyms}
    \label{fig:lastPlot}
\end{figure}

Figure~\ref{fig:lastPlot} shows the amount of data transferred as a function of the total number of revoked pseudonyms. It can be observed that HBFAs provide limited gains over the use of  single BF. In turn, the  implementation based on redactable signatures is highly efficient and allows  clients to perform an authentication by transfering a RL smaller than 1KB.

\begin{figure}[t]
  \centering
  \begin{subfigure}[b]{0.7\columnwidth}
    \includegraphics[width=\columnwidth]{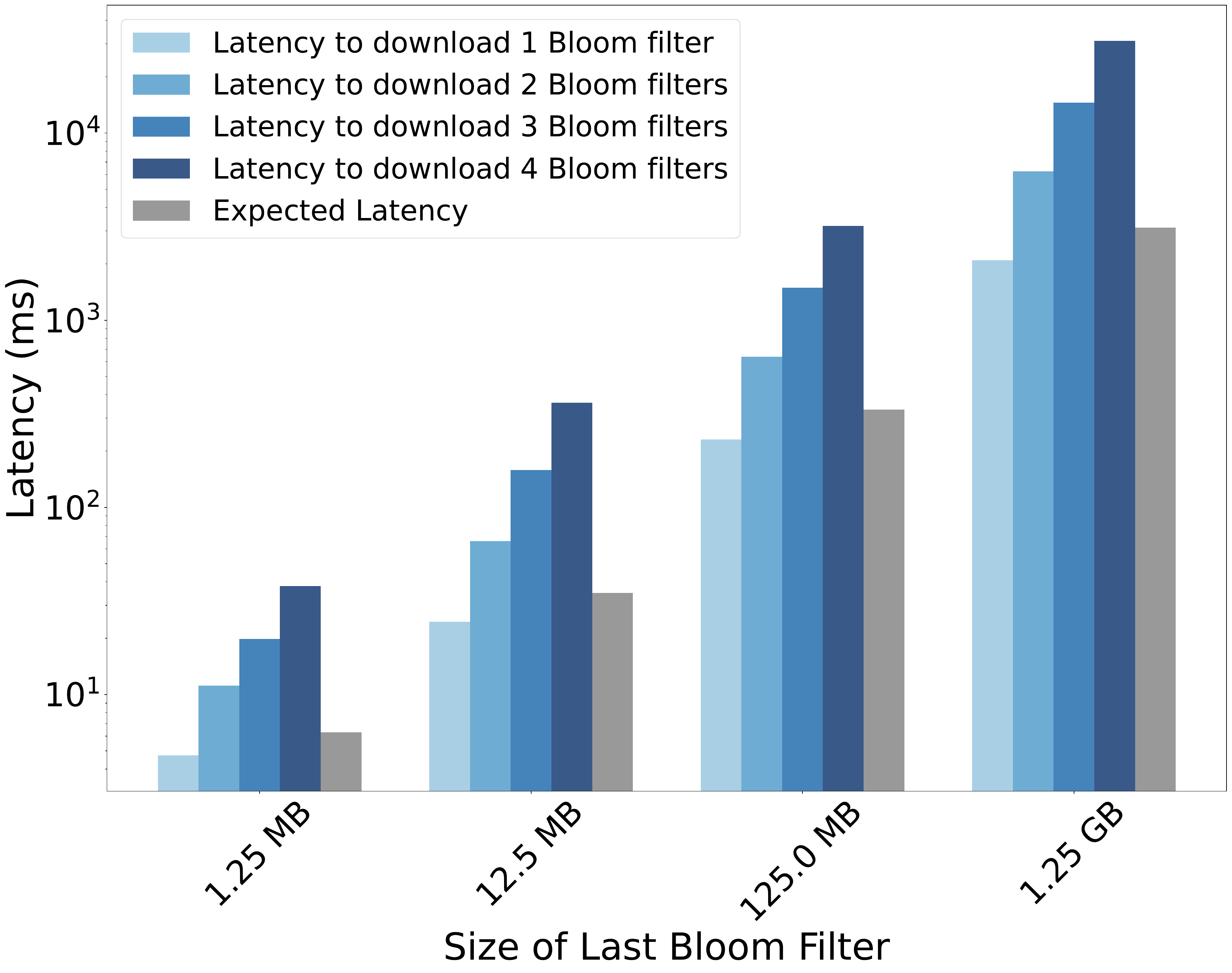}
    \caption{Latency using HBFA}
    \label{fig:cascade_b_filters}
  \end{subfigure}
  %\hfill
  
  \begin{subfigure}[b]{0.7\columnwidth}
    \includegraphics[width=\columnwidth]{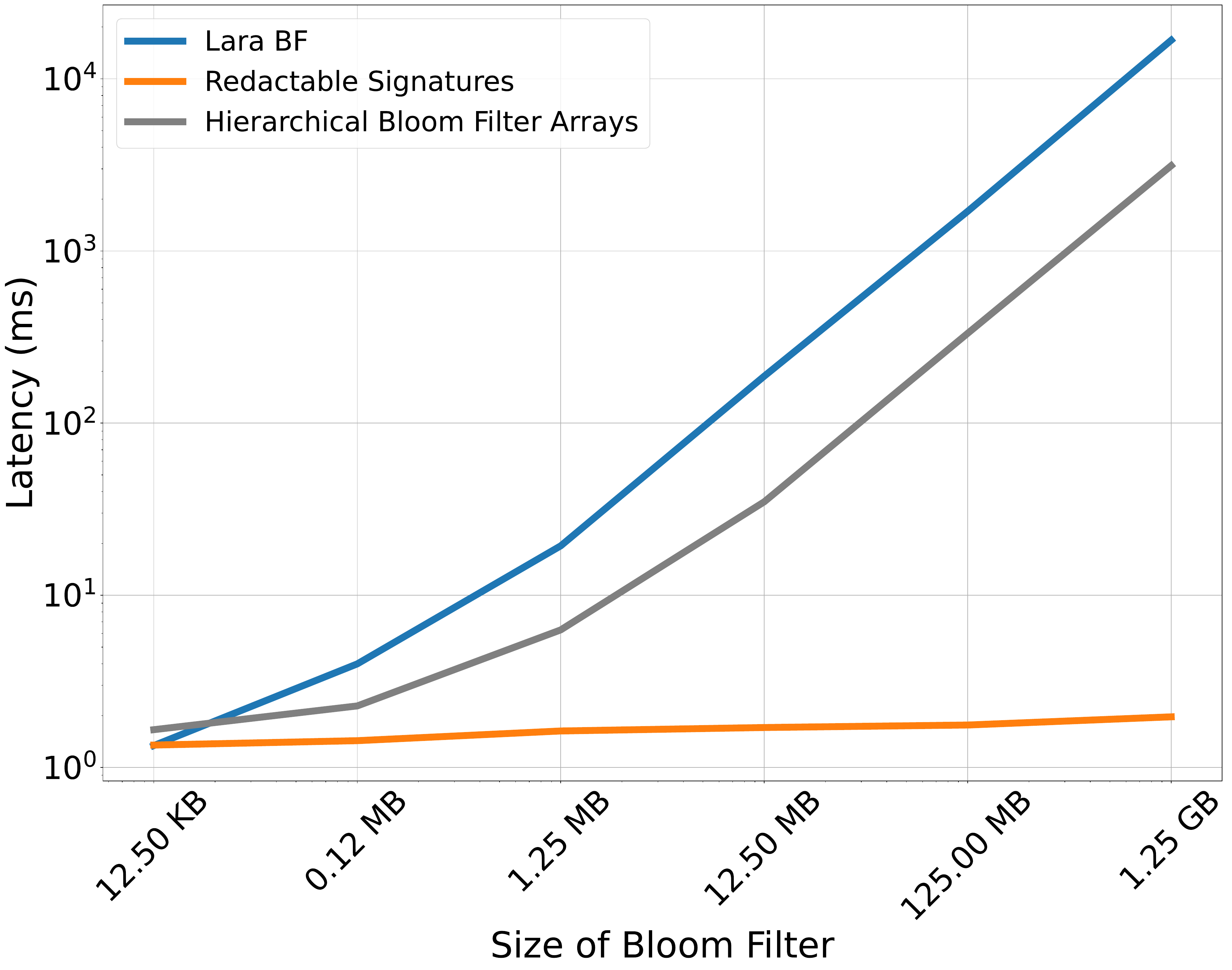}
    \caption{Latency Comparison}
    \label{fig:avg_latencies}
  \end{subfigure}
  \hfill
  \caption{Authentication Latency}
  \label{fig:comparison}
\end{figure}

Figure~\ref{fig:comparison} presents the latency results for the three proposed techniques. In the HBFA implementation, the client retrieves one or more Bloom Filters; we used a configuration with four filters, each growing progressively in size with a specified growth factor of 2. To assess the performance of HBFA, we considered the performance of each individual filter and calculated the expected latency based on the probability of transitioning to the next filter. These results are illustrated in Figure~\ref{fig:cascade_b_filters}. Figure~\ref{fig:avg_latencies} compares the performance of each technique. The single Bloom Filter approach shows a linear increase in latency as the size of the revocation list grows, reaching $\pm 17$ seconds for a 1.25GB list. In contrast, the HBFA configuration follows a similar pattern but with a roughly $80\%$ improvement in efficiency, achieving an authentication latency of $\pm 3$ seconds for the same list. The Redactable Signatures techniques, on the other hand, maintain constant latency regardlessf the list size, with an authentication latency of approximately $\pm 2$ milliseconds.

\section{Conclusions} 

We have introduced \mysystem, a novel lightweight privacy-preserving authentication scheme that ensures backward unlinkability, revocation auditability, and operates independently of timing assumptions. We propose and compare three different implementations of \mysystem, which aim to reduce the amount of data transferred when the revocation audit is performed. The implementation based on Hierarchical Bloom Filter Arrays achieves an 80\% reduction in authentication latency compared to using a single Bloom Filter but adds a noticeable overhead on the time required to create the revocation list. The other implementation, based on Redactable Signatures, proved to be the most efficient: it introduces a negligible overhead when creating the revocation list and enables constant-time authentication, achieving an audit/authentication latency lower than $2$\textit{ms}, regardless of the revocation list size.

\subsection*{Acknowledgements}

This work was supported by the FCT scholarship 2020.05270.BD, by national funds through Funda\c{c}\~ao para a Ci\^encia e a Tecnologia (FCT) via the INESC-ID grant UIDB/50021/2020 and via the SmartRetail project (ref. C6632206063-00466847) financed by IAPMEI, and by the European Union ACES project, 101093126.

\end{document}